\documentclass[11pt]{article}%
\usepackage{amsfonts}
\usepackage{dsfont}
\usepackage{amsmath}
\usepackage{fullpage}
\usepackage[colorlinks=true,linkcolor=blue,citecolor=red,plainpages=false,pdfpagelabels]%
{hyperref}
\usepackage{amssymb}
\usepackage{graphicx}
\usepackage{color}
\usepackage{mathtools}%
\providecommand{\U}[1]{\protect\rule{.1in}{.1in}}
\newtheorem{theorem}{Theorem}

\newtheorem{corollary}[theorem]{Corollary}

\newtheorem{lemma}[theorem]{Lemma}

\newtheorem{remark}[theorem]{Remark}

\newenvironment{proof}[1][Proof]{\noindent\textbf{#1.} }{\ \rule{0.5em}{0.5em}}

\numberwithin{equation}{section}
\begin{document}

\title{\textbf{Approximate reversal of quantum Gaussian dynamics}}
\author{Ludovico Lami\thanks{F\'isica Te\`orica: Informaci\'o i Fen\`omens Qu\`antics,
Departament de F\'isica, Universitat Aut\`onoma de Barcelona, ES-08193
Bellaterra (Barcelona), Spain}
\and Siddhartha Das\thanks{Hearne Institute for Theoretical Physics, Department of
Physics and Astronomy, Louisiana State University, Baton Rouge, Louisiana
70803, USA}
\and Mark M. Wilde \footnotemark[2] \thanks{Center for Computation and Technology,
Louisiana State University, Baton Rouge, Louisiana 70803, USA}}
\date{\today}
\maketitle

\begin{abstract}
Recently, there has been focus on determining the conditions under which the
data processing inequality for quantum relative entropy is satisfied with
approximate equality. The solution of the exact equality case is due to Petz,
who showed that the quantum relative entropy between two quantum states stays
the same after the action of a quantum channel if and only if there is a
\textit{reversal channel} that recovers the original states after the channel
acts. Furthermore, this reversal channel can be constructed explicitly and is
now called the \textit{Petz recovery map}. Recent developments have shown that
a variation of the Petz recovery map works well for recovery in the case of
approximate equality of the data processing inequality. Our main contribution
here is a proof that bosonic Gaussian states and channels possess a particular
closure property, namely, that the Petz recovery map associated to a bosonic
Gaussian state $\sigma$ and a bosonic Gaussian channel $\mathcal{N}$ is itself
a bosonic Gaussian channel. We furthermore give an explicit construction of
the Petz recovery map in this case, in terms of the mean vector and covariance
matrix of the state $\sigma$ and the Gaussian specification of the channel
$\mathcal{N}$.

\end{abstract}

\section{Introduction}

\subsection{Introduction to recoverability in quantum information}

\label{sec:backgr-recover}

Strong subadditivity of quantum entropy is one of the cornerstones of quantum
information theory, on which many fundamental results rely. Defining the
conditional mutual information of a tripartite state $\rho_{ABC}$ as%
\begin{equation}
I(A;B|C)_{\rho}:=S(AC)_{\rho}+S(BC)_{\rho}-S(ABC)_{\rho}-S(C)_{\rho},
\end{equation}
where $S(G)_{\sigma}\equiv-\operatorname{Tr}[\sigma_{G}\log\sigma_{G}]$ is the
quantum entropy of a state $\sigma_{G}$ of a system $G$, strong subadditivity
is equivalent to the non-negativity of conditional mutual
information:\ $I(A;B|C)_{\rho}\geq0$. Initially conjectured in 1967
\cite{robinson67,lanford68}, it was subsequently proven six years later
\cite{lieb73a,lieb73b}. Afterward, its equivalence to the data processing
inequality for the quantum relative entropy~\cite{U62} was realized
\cite{uhlmann73,lindblad74,Lindblad1975,ruskai02}. This latter inequality has
the form
\begin{equation}
D(\rho\Vert\sigma)\geq D(\mathcal{N}(\rho)\Vert\mathcal{N}(\sigma
)),\label{eq:rel-ent-mono}%
\end{equation}
being valid for all states $\rho,\sigma$ and all quantum channels
$\mathcal{N}$ (completely positive, trace-preserving maps). Here, the quantum
relative entropy is defined for quantum states $\rho$ and $\sigma$ as%
\begin{equation}
D(\rho\Vert\sigma)\equiv\operatorname{Tr}[\rho(\log\rho-\log\sigma)],
\end{equation}
whenever the support of $\rho$ is contained in the support of $\sigma$, and it
is set to $+\infty$ otherwise~\cite{U62}.

The interest in strong subadditivity has not fallen over time, and many
different proofs for it have been proposed in the last four decades (see, for
instance,~\cite{nielsen04}). At the same time, new improvements of the original
inequality have recently been found. Extending methods originally proposed in
\cite{effros09}, an operator generalization of strong subadditivity was
recently proven in~\cite{kim12}.

A line of research which is of particular interest to us focuses on
investigating the conditions under which strong subadditivity, or more
generally the data processing inequality for relative entropy, is satisfied
with equality or approximate equality. The solution of the exact equality case
dates back to the 1980s: in~\cite{petz86,petz88,petz03}, it was shown that the
relative entropy between two states stays the same after the action of a
quantum channel if and only if there is a \textit{recovery channel} bringing
back both images to the original states. Furthermore, this reversing channel
can be constructed explicitly and now takes the name \textit{Petz recovery
map}. Afterward,~\cite{Mosonyi2004,M05} proved a structure theorem giving a
form for states and a channel saturating the data-processing inequality for
relative entropy, and, related to this development, the form of tripartite
states satisfying strong subadditivity with equality was determined in
\cite{hayden03}.

Characterising the structure of states for which strong subadditivity is
nearly saturated requires different techniques, and progress was not made
until more recently.
In 2011, a lower bound on conditional mutual information in terms of one-way
LOCC norms~\cite{MWW09} was proven in~\cite{brandao11}, the motivation for
\cite{brandao11} lying in the question of faithfulness of an entanglement
measure called squashed entanglement~\cite{christandl04} (see also
\cite{tucci99,tucci02} for discussions related to squashed entanglement).
Later on, a conjecture put forward in~\cite{winter12} proposed another
operationally meaningful remainder term for the relative entropy decrease
induced by a quantum channel, given by the relative entropy between the state
$\rho$ and a \textquotedblleft recovered version\textquotedblright\ of
$\mathcal{N}(\rho)$. The authors of~\cite{winter12} proposed the following
conjecture as a refinement of~\eqref{eq:rel-ent-mono}:
\begin{equation}
D(\rho\Vert\sigma)\overset{?}{\geq}D(\mathcal{N}(\rho)\Vert\mathcal{N}%
(\sigma))+D(\rho\Vert(\mathcal{R}_{\sigma,\mathcal{N}}\circ\mathcal{N}%
)(\rho))\,, \label{conj VV}%
\end{equation}
where $\mathcal{R}_{\sigma,\mathcal{N}}$ should be a quantum channel depending
only on $\sigma$ and $\mathcal{N}$ and such that $(\mathcal{R}_{\sigma
,\mathcal{N}}\circ\mathcal{N})(\sigma)=\sigma$. The authors of~\cite{winter12}
proved~\eqref{conj VV} in the classical case, when the states $\rho$ and
$\sigma$ commute and the channel is classical as well, and they showed how the
recovery channel in this case can be taken as the Petz recovery map. This
conjecture has now been proven in a number of special, yet physically relevant
cases as well~\cite{AWWW15,BDW16,ML16,LW16,AW16}. Unfortunately, the authors of
\cite{winter12} showed that in the general quantum case, $\mathcal{R}%
_{\sigma,\mathcal{N}}$ in~\eqref{conj VV} cannot be taken as the Petz recovery
map, and most recently, a counterexample to~\eqref{conj VV} has been reported in~\cite{FF17}, so that~\eqref{conj VV} cannot hold generally. For further details, see also~\cite{kim13,li14}, and for related
conjectures, see~\cite{BSW14,SBW14}.

While the general form of the conjecture in~\eqref{conj VV} is not true~\cite{FF17},
in~\cite{fawzi14}, it was shown that if the conditional mutual information
$I(A;B|C)_{\rho}$ is small, then the state $\rho_{ABC}$ can be very well
approximated by one of its \textquotedblleft reconstructed\textquotedblright%
\ versions $\mathcal{R}_{C\rightarrow BC}(\rho_{AC})$. That is, the authors of
\cite{fawzi14} proved the following inequality:
\begin{equation}
I(A;B|C)_{\rho}\,\geq\,-\log F(\rho_{ABC},\mathcal{R}_{C\rightarrow BC}%
(\rho_{AC}))\,, \label{fawzi-renner}%
\end{equation}
where $F$ denotes the quantum fidelity~\cite{U76}, defined as $F(\omega
,\tau):=\Vert\sqrt{\omega}\sqrt{\tau}\Vert_{1}^{2}$ for quantum states
$\omega$ and~$\tau$, and $\mathcal{R}_{C\rightarrow BC}$ is a recovery channel
taking an input system $C$ to output systems $BC$. Furthermore, the channel
$\mathcal{R}_{C\rightarrow BC}$ can be taken as the Petz recovery map up to
some unitary rotations preceding and following its action, but note that the
unitary rotations given in~\cite{fawzi14} generally depend on the full state
$\rho_{ABC}$.

After the result of~\cite{fawzi14} appeared, much activity surrounding entropy
inequalities and recovery channels occurred. An alternative and simpler proof
of the faithfulness of squashed entanglement following the lines of
\cite{winter12} immediately appeared~\cite{li14}, while an alternative proof
of~\eqref{fawzi-renner} that makes use of quantum state redistribution
\cite{DY08,YD09} appeared in~\cite{brandao14}. In~\cite{sutter15a}, an
important particular case of~\eqref{fawzi-renner} was proven; that is, it was
shown that the recovery map in~\eqref{fawzi-renner} can be chosen to depend
only on $\rho_{BC}$ and to obey $\mathcal{R}_{C\rightarrow BC}(\rho_{C}%
)=\rho_{BC}$. A different approach was delivered in~\cite{wilde15}, based on
the methods of complex interpolation~\cite{BL76} and generalized R\'{e}nyi
entropies~\cite{BSW14,SBW14}. The main result of~\cite{wilde15} states that a
lower bound on the decrease in relative entropy induced by a quantum channel
is given by the negative logarithm of the fidelity between the first state and
its recovered version, which is a step closer to the proof of the conjecture
in~\eqref{conj VV}. However, the recovery term in~\cite{wilde15} is weaker
than the right-hand side of~\eqref{conj VV}, and the map appearing in it lacks
one of the two properties that it is required to obey. Another step toward the
proof of the conjecture in~\eqref{conj VV} was performed in~\cite{junge15},
where a more general tool from complex analysis~\cite{hirschman52} and the
methods of~\cite{BSW14,SBW14,wilde15} were exploited in order to prove a
statement similar to~\eqref{conj VV}, with the relative entropy on the
right-hand side substituted by a negative log-fidelity, but with the recovery
map depending only on $\sigma$ and $\mathcal{N}$ and furthermore satisfying
$\mathcal{R}_{\sigma,\mathcal{N}}(\mathcal{N}(\sigma))=\sigma$. Meanwhile, a different proof approach based on pinching was delivered in~\cite{sutter15b}, and then a
systematic method for deriving matrix inequalities by forcing the operators to
commute via the application of suitably chosen \textquotedblleft pinching
maps\textquotedblright\ was proposed in~\cite{sutter16}. This method as well
as the complex interpolation techniques in~\cite{dupuis15} can be also applied
to prove multioperator trace inequalities~\cite{dupuis15, sutter16, wilde16},
which generalise the celebrated Golden-Thompson inequality $\operatorname{Tr}%
[e^{X+Y}]\leq\operatorname{Tr}[e^{X}e^{Y}]$ ($X,Y$ hermitian) and the stronger
statements given in~\cite{lieb73c}. The results of~\cite{sutter16} also marked
further progress toward establishing the conjecture in~\eqref{conj VV}.

\subsection{Introduction to quantum Gaussian states and channels}

A major platform for the application of quantum information theory to physical
information processing is constituted by quantum optics~\cite{GK04} with a
finite number of electromagnetic modes or quantum harmonic oscillators. From
the mathematical perspective, this framework can be thought of as quantum
mechanics applied to separable Hilbert spaces endowed with a finite number of
operators obeying canonical commutation relations~\cite{S17}.

A typical free Hamiltonian of such a system is quadratic in the canonical
operators, and in fact, a special role within this context is played by ground
or thermal states of such Hamiltonians, commonly called \textit{Gaussian
states}. These states define a useful operational framework for several
reasons, stemming from both physics and mathematics~\cite{adesso14,S17}. From
the physical point of view, they are easily produced and manipulated in the
laboratory and can be used to implement effective quantum protocols
\cite{BR04,Wang20071}. Mathematically convenient properties that qualify them
as defining a legitimate framework include

\begin{enumerate}
\item the closure under so-called Gaussian unitary evolutions, that is,
unitaries induced by piecewise time evolution via quadratic Hamiltonians, as
well as more generally

\item the closure under Gaussian channels, which can be understood as the
operation of adding an ancillary system in a vacuum state, applying a global
Gaussian unitary, and tracing out one of the subsystems~\cite{CEGH08}.
\end{enumerate}

\noindent Recently, more advanced \textquotedblleft closure\textquotedblright%
\ properties have been established, such as the optimality of Gaussian states
for optimising the output entropy of one-mode, phase-covariant quantum
channels, even when a fixed value of the input entropy is prescribed
\cite{V13,jack16a,jack16b,jack16c}. These facts have the striking implication
that it suffices to select coding strategies according to Gaussian states in
order to achieve optimal rates in several quantum communication tasks
\cite{GGLMSY04,WHG11,GPCH13,qi16,WQ16,jack16c}.

\subsection{Summary of main result}

The main contribution of our paper is a proof that Gaussian states and
channels possess another closure property: the Petz recovery map associated to
a Gaussian state $\sigma$ and a Gaussian channel $\mathcal{N}$ is itself a
Gaussian channel (see Theorem~\ref{thm:main}). Additionally, we achieve this
result through an explicit construction of the action of such a Gaussian Petz
channel, which lends itself to multiple applications. For instance, with the
formulas we provide, it is possible to construct a counterexample to the inequality in
\eqref{conj VV}, in which all the states and channels involved are Gaussian
and $\mathcal{R}_{\sigma,\mathcal{N}}$ is the Petz recovery map.\footnote{For convenience of the reader, a Mathematica file demonstrating this numerical counterexample is included with our arXiv post \cite{LDW17}.} This is
similar to what happens in the finite-dimensional case. Another application of
our main result is a more explicit form for an entropy inequality from
\cite{junge15}, whenever the states and channel involved are Gaussian.

More broadly, our result has implications for a resource theory of
non-Gaussianity
\cite{PhysRevLett.89.207903,PhysRevA.65.042304,PhysRevA.67.062320,RevModPhys.77.513,G16}%
, which is not currently complete but for which there has been notable
progress. In particular, in such a theory, one takes the free states and free
operations to be quantum Gaussian states and channels, respectively, and the
expensive or resourceful ones to be non-Gaussian. Such an approach is
motivated by concerns from quantum computation using continuous variables, in
which universal quantum computation is enabled only when non-Gaussian
operations are available
\cite{PhysRevLett.89.207903,PhysRevA.65.042304,PhysRevLett.88.097904}, or from
quantum communication theory, in which non-Gaussian operations are needed for
quantum error correction~\cite{NFC09}, for enhancements over classical
communication strategies~\cite{PhysRevA.89.042309,PhysRevA.93.050302}, for
discrimination of coherent states~\cite{PhysRevA.78.022320}, or for effective
quantum repeaters in quantum key distribution~\cite{PhysRevA.90.062316}. One
might expect the Petz recovery channel to play a critical role in a resource
theory of non-Gaussianity as it has in other resource theories
\cite{AWWW15,ML16,LW16}. As such, our result shows that, in such a resource
theory, the Petz recovery channel is a free operation if the state $\sigma$ is
free and the forward channel $\mathcal{N}$ is free as well. One can quantify
non-Gaussianity of a quantum state $\rho$\ via the following information
measure, known as the relative entropy of non-Gaussianity
\cite{PhysRevA.78.060303,PhysRevA.82.052341}:%
\begin{equation}
D_{G}(\rho)\equiv\min_{\sigma\in\mathcal{G}}D(\rho\Vert\sigma)=D(\rho\Vert
\rho_{G}),
\end{equation}
where $\mathcal{G}$ denotes the set of Gaussian states and $\rho_{G}$ denotes
a quantum Gaussian state with the same mean vector and covariance matrix as
$\rho$ (that $\rho_{G}$ is indeed the minimizer was proven in
\cite{PhysRevA.88.012322}). The relative entropy of non-Gaussianity has not
been established as an operationally meaningful quantifier in the
resource-theoretic sense, but one might think it to be the case in light of
the prominence of relative-entropy quantifiers in other resource theories
\cite{PhysRevLett.115.070503}. However, if it eventually is, our work combined
with the main result of~\cite{junge15}\ would be relevant, given that these
results establish the following interesting inequality, holding for an
arbitrary quantum state $\rho$ and quantum Gaussian channel $\mathcal{N}_{G}$:%
\begin{equation}
D_{G}(\rho)\geq D_{G}(\mathcal{N}_{G}(\rho))-\int_{\mathbb{R}}dt\ p(t)\log
F(\rho,(\mathcal{P}_{\rho_{G},\mathcal{N}_{G}}^{t/2}\circ\mathcal{N}_{G}%
)(\rho)), \label{eq:G-entropy-ineq-recovery}%
\end{equation}
where $p(t):=\frac{\pi}{2}(\cosh(\pi t)+1)^{-1}$ is a probability distribution
parametrized by $t\in\mathbb{R}$ and $\mathcal{P}_{\rho_{G},\mathcal{N}_{G}%
}^{t}$ is a rotated Petz channel~\cite{wilde15}. A corollary of our main
result is that $\mathcal{P}_{\rho_{G},\mathcal{N}_{G}}^{t}$ is a quantum
Gaussian channel (Corollary~\ref{cor:rotated-petz-form}). The inequality in
\eqref{eq:G-entropy-ineq-recovery} has an interpretation similar to that in
previous works:\ if the relative entropy of non-Gaussianity does not decrease
too much under the action of a free operation $\mathcal{N}_{G}$ (so that
$D_{G}(\rho)\approx D(\mathcal{N}_{G}(\rho))$), then one can approximately
reverse the action of $\mathcal{N}_{G}$ by employing a free operation
$\mathcal{P}_{\rho_{G},\mathcal{N}_{G}}^{t/2}$ chosen randomly according to
$p(t)$. Note that one can also write the inequality above as follows:%
\begin{equation}
D(\rho\Vert\rho_{G})\geq D(\mathcal{N}_{G}(\rho)\Vert\mathcal{N}_{G}(\rho
_{G}))-\int_{\mathbb{R}}dt\ p(t)\log F(\rho,(\mathcal{P}_{\rho_{G}%
,\mathcal{N}_{G}}^{t/2}\circ\mathcal{N}_{G})(\rho)).
\label{eq:G-entropy-ineq-recovery-1}%
\end{equation}

We should note that the inequalities stated above are not in contradiction
with the well known no-go theorem for Gaussian quantum error correction
\cite{NFC09}. The main result of~\cite{NFC09} is the following statement:\ if
one is trying to use a Gaussian quantum channel to distill entanglement
between spatially separated parties, then Gaussian encodings combined with
Gaussian decodings are not helpful for this task, whenever performance is
measured with respect to an entanglement measure called logarithmic
negativity. In the inequalities in
\eqref{eq:G-entropy-ineq-recovery}--\eqref{eq:G-entropy-ineq-recovery-1}, the
recovery channel is indeed a quantum Gaussian channel, but the only statement
that these inequalities make is that the performance of the Gaussian Petz
recovery channel for recovery is limited by the relative entropy difference
$D(\rho\Vert\rho_{G})-D(\mathcal{N}_{G}(\rho)\Vert\mathcal{N}_{G}(\rho_{G}))$.

Finally, we suspect that our main result about Gaussian Petz channels might be
useful in contexts beyond the traditional ones in quantum information theory.
Indeed, Petz recovery maps have recently been employed in the context of
high-energy physics, quantum many-body physics, and topological order
\cite{SM16,ZS16,K16}, and so our result here could be useful if the states
involved in those contexts are Gaussian states.


This paper is structured as follows. In Section~\ref{sec:background}, we
review some background material and establish notation. In particular, we
review the Petz recovery map (Section~\ref{sec:petz-recovery-review}) and
bosonic Gaussian states and channels (Section~\ref{sec:Gaussian-review}). In
Section~\ref{sec:main-result}, we state our main result,
Theorem~\ref{thm:main}, which establishes that the Petz recovery map for a
Gaussian state $\sigma$ and a Gaussian channel~$\mathcal{N}$ is itself a
Gaussian channel, and we give an explicit form for it in terms of the
parameters that characterize $\sigma$ and $\mathcal{N}$.
Corollary~\ref{cor:rotated-petz-form} establishes a similar result for the
rotated Petz maps from~\cite{wilde15}. Our proof of Theorem~\ref{thm:main} is
divided into four parts, given in Sections~\ref{sec:first-step-zero-mean-etc}%
--\ref{sec rigorous}. We conclude in Section~\ref{sec:conclusion} with a
summary and some open questions. We point the interested reader to
Appendix~\ref{golden}, in which we give a method for computing products of
exponentials of inhomogeneous quadratic Hamiltonians, building upon
\cite{Balian1969}. Although results of~\cite{petz86,petz88,PETZ}\ establish
that the Petz map is completely positive and trace-preserving,
Appendix~\ref{sec:Petz-CP}\ offers a different argument that the Gaussian Petz
map is completely positive.

\section{Background and notation}

\label{sec:background}

\subsection{Petz recovery map}

\label{sec:petz-recovery-review}

As discussed in Section~\ref{sec:backgr-recover}, the Petz recovery map is a
notable object playing a crucial role in the theory of quantum recoverability.
It has been interpreted in~\cite{PhysRevA.88.052130} as a quantum
generalization of the Bayes rule from probability theory. Given a state
$\sigma$ and a channel $\mathcal{N}$, the associated Petz map $\mathcal{P}%
_{\sigma,\mathcal{N}}$ is defined as a linear map satisfying the following
\cite{petz86,petz88,PETZ}:%
\begin{equation}
\langle A,\mathcal{N}^{\dag}(B)\rangle_{\sigma}=\langle\mathcal{P}%
_{\sigma,\mathcal{N}}^{\dag}(A),B\rangle_{\mathcal{N}(\sigma)},\ \ \ \forall
A,B,\label{eq:Petz-equations}%
\end{equation}
where $A$ and $B$ are bounded operators and the weighted Hilbert--Schmidt
inner product is defined for bounded operators $\tau_{1}$ and $\tau_{2}$ and a
trace-class operator $\xi$ as%
\begin{equation}
\left\langle \tau_{1},\tau_{2}\right\rangle _{\xi}\equiv\operatorname{Tr}%
[\tau_{1}^{\dag}\xi^{1/2}\tau_{2}\xi^{1/2}].
\end{equation}
The map $\mathcal{P}_{\sigma,\mathcal{N}}$ is unique if $\mathcal{N}(\sigma)$
is a faithful operator~\cite{petz86,petz88,PETZ}, and otherwise, it is unique
on the support of this operator. If $\sigma$ acts on a finite-dimensional
Hilbert space and $\mathcal{N}$ is a quantum channel with finite-dimensional
inputs and outputs, then the Petz map takes the following explicit form
\cite{hayden03}:%
\begin{equation}
\mathcal{P}_{\sigma,\mathcal{N}}(\omega)\,\equiv\,\sigma^{1/2}\mathcal{N}%
^{\dag}\!\left(  \mathcal{N}(\sigma)^{-1/2}\omega\mathcal{N}(\sigma
)^{-1/2}\right)  \sigma^{1/2}\,,\label{Petz}%
\end{equation}
where $\mathcal{N}(\sigma)^{-1/2}$ is understood as a generalized inverse
(i.e., inverse on the support of $\mathcal{N}(\sigma)$). Sometimes we omit the
dependence of $\mathcal{P}$ on $\sigma$ and $\mathcal{N}$ for the sake of
simplicity. A rotated Petz map $\mathcal{P}_{\sigma,\mathcal{N}}^{t}$ for
$t\in\mathbb{R}$, a state $\sigma$, and a channel $\mathcal{N}$ is defined as
\cite{wilde15}%
\begin{equation}
\mathcal{P}_{\sigma,\mathcal{N}}^{t}(\omega)\equiv\sigma^{it}\mathcal{P}%
_{\sigma,\mathcal{N}}(\mathcal{N}(\sigma)^{-it}\omega\mathcal{N}(\sigma
)^{it})\sigma^{-it},\label{eq:rotated-petz}%
\end{equation}
with $\sigma^{it}=\exp(it\log\sigma)$ being understood as a unitary evolution
according to the Hamiltonian $\log\sigma$. Even if it is not \textit{a priori} apparent, it can be shown that the Petz map~\eqref{Petz} as well as its rotated versions~\eqref{eq:rotated-petz} are guaranteed to be completely positive and trace-preserving (i.e., valid quantum channels) under the above hypotheses.

\subsection{Quantum Gaussian states and channels}

\label{sec:Gaussian-review}

Here we provide some background on bosonic quantum Gaussian states and channels (see
\cite{CEGH08,adesso14,S17}\ for reviews). An $n$-mode quantum system is
described by a density operator acting on a tensor-product Hilbert space. To
the $j$th Hilbert space in the tensor product, for $j\in\left\{
1,\ldots,n\right\}  $, we let $x_{j}$ and $p_{j}$ denote the position- and
momentum-quadrature operator, respectively. These operators satisfy the
canonical commutation relations:\ $\left[  x_{j},p_{k}\right]  =i\delta_{j,k}%
$, where we have set $\hbar=1$. It is convenient to form a vector
$r=(x_{1},\ldots,x_{n},p_{1},\ldots,p_{n})^{T}$\ from these operators, and
then we can rewrite the canonical commutation relations in matrix form as
follows:%
\begin{equation}
\lbrack r,r^{T}]=i\Omega,\label{CCR}%
\end{equation}
where%
\begin{equation}
\Omega\equiv%
\begin{bmatrix}
0 & 1\\
-1 & 0
\end{bmatrix}
\otimes I_{n},
\end{equation}
and $I_{n}$ denotes the $n\times n$ identity matrix. We often make use of the
identities $\Omega^{T}\Omega=I$ and $\Omega^{T}=-\Omega$.

The displacement (Weyl)\ operator $D_{z}$ plays an important role in Gaussian
quantum information, defined for $z\in\mathbb{R}^{2n}$ as%
\begin{equation}
D_{z}\equiv\exp(iz^{T}\Omega r).
\end{equation}
For $z_{1},z_{2}\in\mathbb{R}^{2n}$, the displacement operators satisfy the
following composition rule:%
\begin{equation}
D_{z_{1}}D_{z_{2}}=D_{z_{1}+z_{2}}e^{-\frac{i}{2}z_{1}^{T}\Omega z_{2}%
}.\label{comp displacement}%
\end{equation}
It can be shown that displacement operators form a complete, orthogonal set of
operators, and their Hilbert--Schmidt orthogonality relation is as follows:
\begin{equation}
\operatorname{Tr}[D_{z_{1}}D_{-z_{2}}]=(2\pi)^{n}\delta(z_{1}-z_{2}%
).\label{trace displacement}%
\end{equation}
Moreover, due to their completeness, these operators allow for a Fourier-Weyl
expansion of a quantum state, in terms of a characteristic function. In more
detail, a quantum state $\rho$\ has a characteristic function $\chi_{\rho}%
(w)$,\ defined as%
\begin{equation}
\chi_{\rho}(w)\equiv\operatorname{Tr}[\rho D_{-w}],
\end{equation}
and the original state $\rho$ can be written in terms of $\chi_{\rho}(w)$ as%
\begin{equation}
\rho=\int\frac{d^{2n}w}{\left(  2\pi\right)  ^{n}}\ \chi_{\rho}(w)\ D_{w}.
\end{equation}
The mean vector $s_{\rho}\in\mathbb{R}^{2n}$ and $2n\times2n$ covariance
matrix $V_{\rho}$ of a quantum state $\rho$ are defined as
\begin{align}
s_{\rho} &  \equiv\left\langle r\right\rangle _{\rho}=\operatorname{Tr}%
[r\rho],\\
V_{\rho} &  \equiv\langle\{r-s_{\rho},r^{T}-s_{\rho}^{T}\}\rangle_{\rho
}=\operatorname{Tr}[\{r-s_{\rho},r^{T}-s_{\rho}^{T}\}\rho].
\end{align}
It follows from the above definition that the covariance matrix $V_{\rho}$ is symmetric.

A quantum Gaussian state is a ground or thermal state of a Hamiltonian that is
quadratic in the position- and momentum-quadrature operators. In particular,
up to an irrelevant additive constant, any such Hamiltonian has the form $\frac{1}{2}\left(  r-s\right)  ^{T}H\left(
r-s\right)  $, where $s\in\mathbb{R}^{2n}$ and $H$ is a $2n\times2n$ positive
definite matrix that we refer to as the Hamiltonian matrix. Then a quantum
Gaussian state $\rho$ takes the form%
\begin{equation}
\rho=Z_{\rho}^{-1}\exp\!\left(  -\frac{1}{2}(r-s_{\rho})^{T}H_{\rho}%
(r-s_{\rho})\right)  , \label{eq:exp-form-Gaussian}%
\end{equation}
where $Z_{\rho}\equiv\operatorname{Tr}\!\left[  \exp\!\left(  -\frac{1}%
{2}(r-s_{\rho})^{T}H_{\rho}(r-s_{\rho})\right)  \right]  $ and one can show
that $\left\langle r\right\rangle _{\rho}=s_{\rho}\in\mathbb{R}^{2n}$ (i.e.,
$s_{\rho}$ is the mean vector of $\rho$). Defining%
\begin{equation}
V_{\rho}\equiv\coth\!\left(  \frac{i\Omega H_{\rho}}{2}\right)  i\Omega,
\label{V coth}%
\end{equation}
one can also show that $V_{\rho}$ is the covariance matrix of $\rho$ \cite{PhysRevA.71.062320,K06,H11,Holevo2011,H12}, whose
matrix elements satisfy $V_{\rho}^{j,k}=\langle\{r_{j}-s_{\rho}^{j}%
,r_{k}-s_{\rho}^{k}\}\rangle_{\rho}$ and the Heisenberg uncertainty relation
\cite{PhysRevA.49.1567}:%
\begin{equation}
V_{\rho}+i\Omega\geq0.
\end{equation}
A quantum Gaussian state is faithful (having full support)\ if $V_{\rho
}+i\Omega>0$.

A quantum Gaussian state $\rho$\ with mean vector $s_{\rho}$ and covariance
matrix $V_{\rho}$ has the following Gaussian characteristic function:%
\begin{equation}
\chi_{\rho}(w)=\exp\!\left(  -\frac{1}{4}\left(  \Omega w\right)  ^{T}V_{\rho
}\Omega w+i\left(  \Omega w\right)  ^{T}s_{\rho}\right)  ,
\label{char funct G}%
\end{equation}
so that it can be written in the following way:%
\begin{equation}
\rho=\int\frac{d^{2n}w}{(2\pi)^{n}}\ \exp\!\left(  -\frac{1}{4}\left(  \Omega
w\right)  ^{T}V_{\rho}\Omega w+i\left(  \Omega w\right)  ^{T}s_{\rho}\right)
D_{w}. \label{G state}%
\end{equation}
After a change of variables ($w\rightarrow\Omega w$), this representation
becomes%
\begin{equation}
\rho=\int\frac{d^{2n}w}{(2\pi)^{n}}\ \exp\!\left(  -\frac{1}{4}w^{T}V_{\rho
}w-iw^{T}s_{\rho}\right)  D_{\Omega w}. \label{G state 2}%
\end{equation}

A quantum Gaussian channel is a completely positive, trace-preserving map that
takes Gaussian input states to Gaussian output states. A quantum Gaussian
channel $\mathcal{N}$ that takes $n$-mode Gaussian input states to $m$-mode
Gaussian output states is specified by a $2m\times2n$ transformation matrix
$X$, a $2m\times2m$ positive semi-definite, additive noise matrix $Y$, and a
displacement vector $\delta\in\mathbb{R}^{2n}$. The action of such a channel
on a generic state $\rho$ with characteristic function $\chi_{\rho}(w)$ is to
output a state $\mathcal{N}(\rho)$ having the following characteristic
function:%
\begin{equation}
\chi_{\mathcal{N}(\rho)}(w)=\chi_{\rho}(\Omega^{T}X^{T}\Omega w)\exp\!\left(
-\frac{1}{4}\left(  \Omega w\right)  ^{T}Y\Omega w+i\left(  \Omega w\right)
^{T}\delta\right)  . \label{eq:Gaussian-channel-char-func}%
\end{equation}
Then the channel $\mathcal{N}$ leads to the following transformation of the
covariance matrix $V$ and mean vector $s$ of an input quantum Gaussian state:%
\begin{equation}
\mathcal{N}:\ \left\{
\begin{array}
[c]{lcl}%
V & \longmapsto & XVX^{T}+Y\\
s & \longmapsto & Xs+\delta
\end{array}
\right.  . \label{N cm}%
\end{equation}
The matrices $X$ and $Y$ should satisfy the following condition in order for
the map $\mathcal{N}$ to be completely positive:%
\begin{equation}
Y+i\Omega\geq iX\Omega X^{T}.
\label{G cp condition}
\end{equation}

The adjoint of a quantum channel $\mathcal{N}$ is defined as the unique linear
map satisfying the following for all $A$ and $B$:%
\begin{equation}
\langle A,\mathcal{N}(B)\rangle=\langle\mathcal{N}^{\dag}(A),B\rangle,
\end{equation}
where $B$ is an arbitrary trace-class operator, $A$ is an arbitrary bounded
operator, and the Hilbert--Schmidt inner product is defined for operators
$A_{1}$ and $A_{2}$ as $\left\langle A_{1},A_{2}\right\rangle \equiv
\operatorname{Tr}[A_{1}^{\dag}A_{2}]$. The adjoint map $\mathcal{N}^{\dag}$ is
completely positive and unital if $\mathcal{N}$ is completely positive and
trace-preserving. The action of the adjoint $\mathcal{N}^{\dag}$ of a quantum
Gaussian channel $\mathcal{N}$ defined by~\eqref{N cm}\ is as follows
\cite{CEGH08,GLS16}, when acting on a displacement operator $D_{\Omega z}$:%
\begin{equation}
\mathcal{N}^{\dag}(D_{\Omega z})=D_{\Omega X^{T}z}\exp\!\left(  -\frac{1}%
{4}z^{T}Yz+iz^{T}\delta\right)  . \label{eq:adjoint-on-disps}%
\end{equation}
The action of the adjoint $\mathcal{N}^{\dag}$ on a quantum Gaussian state
with covariance matrix $V$\ and mean vector $s$ is then to output a quantum
Gaussian operator described by covariance matrix $X^{-1}\left(  V+Y\right)
X^{-T}$ and mean vector $X^{-1}(s-\delta)$ whenever $X$ is invertible
\cite[Appendix~B]{GLS16}. We summarize these transformation rules as follows:%
\begin{equation}
\mathcal{N}^{\dag}:\ \left\{
\begin{array}
[c]{lcl}%
V & \longmapsto & X^{-1}\left(  V+Y\right)  X^{-T}\\
s & \longmapsto & X^{-1}(s-\delta)
\end{array}
\right.  . \label{N dag cm}%
\end{equation}
Typically one thinks of the channel $\mathcal{N}$ as acting in the
Schr\"{o}dinger picture, taking input states to output states, and one thinks
of the adjoint $\mathcal{N}^{\dag}$ as acting in the Heisenberg picture,
taking input bounded operators to output bounded operators. So this is why we
have specified the channel $\mathcal{N}$ in terms of its action on
characteristic functions, which describe states, and the adjoint
$\mathcal{N}^{\dag}$ in terms of its action on displacement operators, a
natural choice of bounded operators in our context here.

Often we find it useful to write%
\begin{equation}
\sigma=D_{s_{\sigma}}^{\dag}\sigma_{0}D_{s_{\sigma}},
\label{eq:state-means-out}%
\end{equation}
where $\sigma_{0}$ is a Gaussian state with the same covariance matrix as
$\sigma$ but with vanishing mean vector. Analogously, the channel
$\mathcal{N}$ in~\eqref{eq:Gaussian-channel-char-func}\ admits the following
decomposition:%
\begin{equation}
\mathcal{N}(\cdot)=D_{\delta}^{\dag}\mathcal{N}_{0}(\cdot)D_{\delta},
\label{eq:channel-means-out}%
\end{equation}
where $\mathcal{N}_{0}$ is a zero-displacement Gaussian channel, acting as in
\eqref{N cm} but with $\delta=0$. Taking the adjoint gives%
\begin{equation}
\mathcal{N}^{\dag}(\cdot)=\mathcal{N}_{0}^{\dag}\!\left(  D_{\delta}%
(\cdot)D_{\delta}^{\dag}\right)  . \label{eq:adjoint-ch-means-out}%
\end{equation}
Applying $\mathcal{N}$ to $\sigma$ yields%
\begin{equation}
\mathcal{N}(\sigma)=D_{Xs+\delta}^{\dag}\mathcal{N}_{0}(\sigma_{0}%
)D_{Xs+\delta}, \label{eq:channel-state-means-out}%
\end{equation}
which follows from~\eqref{N cm}. We also make use of the following channel
covariance relations:%
\begin{align}
\mathcal{N}(D_{\gamma}^{\dag}(\cdot)D_{\gamma})  &  =D_{X\gamma+\delta}^{\dag
}\mathcal{N}_{0}(\cdot)D_{X\gamma+\delta},\label{eq:channel-cov-disps}\\
\mathcal{N}^{\dag}(D_{\gamma}^{\dag}(\cdot)D_{\gamma})  &  =D_{X^{-1}%
(\gamma-\delta)}^{\dag}\mathcal{N}_{0}^{\dag}(\cdot)D_{X^{-1}(\gamma-\delta)},
\label{eq:adjoint-channel-cov-disps}%
\end{align}
which follow from~\eqref{eq:Gaussian-channel-char-func},~\eqref{N cm},
\eqref{eq:adjoint-on-disps}, and~\eqref{N dag cm}. Note that
\eqref{eq:adjoint-channel-cov-disps}\ holds whenever $X$ is invertible.

Finally, given a Gaussian state $\sigma$ with mean vector $s_{\sigma}$ and
covariance matrix $V_{\sigma}$, we can consider a unitary rotation of the form
$\sigma^{it}=\exp(it\log\sigma)$ for $t\in\mathbb{R}$. By using the
representation in~\eqref{eq:exp-form-Gaussian} with the Hamiltonian matrix
$H_{\sigma}$, we can write the unitary $\sigma^{it}$\ as%
\begin{align}
\sigma^{it}  &  =\exp\!\left(  -\frac{i}{2}\left(  r-s_{\sigma}\right)
^{T}H_{\sigma}t\left(  r-s_{\sigma}\right)  \right)  \exp\!\left(  -it\log
Z_{\sigma}\right) \label{eq:sigma-to-it}\\
&  =D_{-s_{\sigma}}\left[  \exp\!\left(  \frac{i}{2}r^{T}\left(  -H_{\sigma
}t\right)  r\right)  \exp\!\left(  -it\log Z_{\sigma}\right)  \right]
D_{s_{\sigma}},
\end{align}
where we have used the fact that $\left(  r-s_{\sigma}\right)  ^{T}H_{\sigma
}\left(  r-s_{\sigma}\right)  =D_{-s_{\sigma}}r^{T}H_{\sigma}rD_{s_{\sigma}}$
and the operator identity $B\exp(A)B^{-1}=\exp(BAB^{-1})$. The unitary
$\sigma^{it}$ is a Gaussian unitary because it is generated by a Hamiltonian
no more than quadratic in the position- and momentum-quadrature operators. Let us define the
symplectic transformation corresponding to the unitary $\exp\!\left(  \frac
{i}{2}r^{T}\left(  -H_{\sigma}t\right)  r\right)  $ as%
\begin{equation}
S_{\sigma,t}\equiv\exp(\Omega H_{\sigma}t),
\end{equation}
so that%
\begin{equation}
\sigma^{it}r\sigma^{-it}=S_{\sigma,-t}\left(  r-s_{\sigma}\right)  + s_{\sigma},
\label{eq:sigma-to-it-sympaction}%
\end{equation}
where we used that $D_{s_{\sigma}}rD_{-s_{\sigma}}=r+s_{\sigma}$. The above formula implies that
\begin{align}
V_{\sigma^{it} \omega \sigma^{-it}} &= S_{\sigma, t} V_\omega S_{\sigma, t}^T , \\
s_{\sigma^{it} \omega \sigma^{-it}} &= S_{\sigma,t} (s_\omega - s_\sigma) + s_\sigma .
\end{align}

\section{Main result: Petz map as a quantum Gaussian channel}

\label{sec:main-result}

Our main result is the following theorem:

\begin{theorem}
\label{thm:main}Let $\sigma$ be a quantum Gaussian state with mean vector
$s_{\sigma}$ and covariance matrix $V_{\sigma}$, and let $\mathcal{N}$ be a
quantum Gaussian channel with its action on an input state as described in
\eqref{N cm}. Suppose furthermore that $\mathcal{N}(\sigma)$ is a faithful
quantum state. Then the Petz recovery map $\mathcal{P}_{\sigma,\mathcal{N}}$
is a quantum Gaussian channel with the following action:%
\begin{equation}
\mathcal{P}_{\sigma,\mathcal{N}}:\ \left\{
\begin{array}
[c]{lcl}%
V & \longmapsto & X_{P}VX_{P}^{T}+Y_{P}\\
s & \longmapsto & X_{P}s+\delta_{P}%
\end{array}
\right.  , \label{eq:Petz-Gaussian}%
\end{equation}
where%
\begin{align}
X_{P}  &  \equiv\sqrt{I+\left(  V_{\sigma}\Omega\right)  ^{-2}}V_{\sigma}%
X^{T}\sqrt{I+\left(  \Omega V_{\mathcal{N}(\sigma)}\right)  ^{-2}}%
^{-1}V_{\mathcal{N}(\sigma)}^{-1}, \label{eq:XP} \\
Y_{P}  &  \equiv V_{\sigma}-X_{P}V_{\mathcal{N}(\sigma)}X_{P}^{T}, \label{eq:YP} \\
\delta_{P}  &  \equiv s_{\sigma}-X_{P}\left(  Xs_{\sigma}+\delta\right)
,\label{eq:Petz-displacement}\\
V_{\mathcal{N}(\sigma)}  &  =XV_{\sigma}X^{T}+Y.
\end{align}
That is, $\mathcal{P}_{\sigma,\mathcal{N}}$ in~\eqref{eq:Petz-Gaussian}\ is
the unique linear map satisfying~\eqref{eq:Petz-equations} for $\sigma$ and
$\mathcal{N}$ as described above.
\end{theorem}

It may not be obvious \textit{a priori}, but the Gaussian map defined by~\eqref{eq:XP} and~\eqref{eq:YP} is indeed a valid Gaussian channel; i.e., it meets the requirement given by~\eqref{G cp condition}. An explicit verification of this fact is provided in Appendix~\ref{sec:Petz-CP}.
The following corollary is a direct consequence of Theorem~\ref{thm:main} and
the discussion surrounding~\eqref{eq:sigma-to-it}--\eqref{eq:sigma-to-it-sympaction}:

\begin{corollary}
\label{cor:rotated-petz-form}For $\sigma$ and $\mathcal{N}$ as given in
Theorem~\ref{thm:main}, the rotated Petz map $\mathcal{P}_{\sigma,\mathcal{N}%
}^{t}$\ (defined in~\eqref{eq:rotated-petz}) is also a quantum Gaussian
channel with the same action as the Petz recovery channel $\mathcal{P}%
_{\sigma,\mathcal{N}}$ but with the substitutions%
\begin{align}
X_{P}  &  \rightarrow X_{P}^{t} \equiv S_{\sigma,t}X_{P}S_{\mathcal{N}%
(\sigma),-t},\\
Y_{P}  &  \rightarrow Y_{P}^{t} \equiv S_{\sigma,t}Y_{P}S_{\sigma,t}^{T},\\
\delta_{P}  &  \rightarrow\delta_{P}^{t}\equiv s_{\sigma}-X_{P}^{t}\left(
Xs_{\sigma}+\delta\right)  .
\end{align}
That is, $\mathcal{P}_{\sigma,\mathcal{N}}^{t}$ is a quantum Gaussian channel
with the following action:%
\begin{equation}
\mathcal{P}_{\sigma,\mathcal{N}}^{t}:\ \left\{
\begin{array}
[c]{lcl}%
V & \longmapsto & X_{P}^{t}V(X_{P}^{t})^{T}+Y_{P}^{t}\\
s & \longmapsto & X_{P}^{t}s+\delta_{P}^{t}%
\end{array}
\right.  .
\end{equation}

\end{corollary}

\begin{remark}
The following entropy inequality was proven to hold whenever $\rho$ and
$\sigma$ are density operators and $\mathcal{N}$ is a quantum channel
\cite{junge15}:%
\begin{equation}
D(\rho\Vert\sigma)\geq D(\mathcal{N}(\rho)\Vert\mathcal{N}(\sigma
))-\int_{\mathbb{R}}dt\ p(t)\log F(\rho,(\mathcal{P}_{\sigma,\mathcal{N}%
}^{t/2}\circ\mathcal{N})(\rho)),
\end{equation}
where $p(t):=\frac{\pi}{2}(\cosh(\pi t)+1)^{-1}$ is a probability distribution
parametrized by $t\in\mathbb{R}$. In the case that $\rho$ and $\sigma$ are
quantum Gaussian states and $\mathcal{N}$ is a quantum Gaussian channel,
Corollary~\ref{cor:rotated-petz-form}\ allows us to conclude that
$\mathcal{P}_{\sigma,\mathcal{N}}^{t/2}$ is a quantum Gaussian channel for all
$t\in\mathbb{R}$. Furthermore, there are explicit, compact formulas for the
relative entropy~\cite{PhysRevA.64.063811,PhysRevA.71.062320,K06,PLOB15}\ and
fidelity~\cite{PS00,WKO00,MM12,Pirandola}\ of two quantum Gaussian states. In both cases, the
formulas are given exclusively in terms of the mean vectors and covariance
matrices of the involved states. Thus, when the states and channel involved
are all Gaussian, the above inequality can be rewritten in a simpler form involving only finite-dimensional matrices instead of trace-class operators acting on infinite-dimensional Hilbert spaces.
\end{remark}

The forthcoming subsections establish a proof of Theorem~\ref{thm:main}.
Before delving into our proof, we highlight our proof strategy, which proceeds
according to the following steps:

\begin{enumerate}
\item \textit{Even though the explicit form of the Petz map in
\eqref{Petz}\ is not generally valid in the infinite-dimensional case because
the inverse of a density operator may be unbounded, we work with it anyway, as
an ansatz} (call this \textbf{Ansatz~1}). Under Ansatz~1, we first show that
it suffices to consider the case when the state $\sigma$ is a zero-mean
Gaussian state and the channel $\mathcal{N}$ does not apply any displacement
to the mean vector of its input, so that $s_{\sigma}=0$ and $\delta=0$, with
$\delta$ defined in~\eqref{eq:Gaussian-channel-char-func} and~\eqref{N cm}.

\item Under the same Ansatz~1, we arrive at the hypothesis that
\eqref{eq:Petz-Gaussian} gives the explicit form for the action of the Petz
map on a Gaussian input state. Recall from~\eqref{Petz} that the Petz map is a
serial concatenation of three completely positive maps:%
\begin{align}
(\cdot) &  \rightarrow\mathcal{N}(\sigma)^{-1/2}(\cdot)\mathcal{N}%
(\sigma)^{-1/2},\label{eq:Petz-1}\\
(\cdot) &  \rightarrow\mathcal{N}^{\dag}(\cdot),\label{eq:Petz-2}\\
(\cdot) &  \rightarrow\sigma^{1/2}(\cdot)\sigma^{1/2}.\label{eq:Petz-3}%
\end{align}
To handle the first completely positive\ map in~\eqref{eq:Petz-1}, we proceed
with an additional ansatz (\textbf{Ansatz~2}) \textit{that taking the inverse
of a Gaussian state corresponds to negating its covariance matrix}. This is
motivated by the representation in~\eqref{eq:exp-form-Gaussian}, in which
inverting the density operator has the effect of negating the Hamiltonian
matrix, which in turn has the effect of negating the covariance matrix due to
the fact that $\operatorname{arcoth}$ is an odd function. Furthermore, results
of~\cite{PS00} allow us to conclude that sandwiching a
Gaussian state by the square root of another Gaussian state is a Gaussian map
resulting in another unnormalized, Gaussian state. To handle the second map in
\eqref{eq:Petz-2}, we can directly apply a result given in~\cite[Appendix~B]%
{GLS16}, which gives an explicit form for the action of the adjoint of a
Gaussian channel on a Gaussian state (see also the review in
\eqref{N dag cm}). We also work with a final \textbf{Ansatz~3}, \textit{which
is the assumption that the matrix }$X$\textit{ in~\eqref{N cm}\ is
invertible}. Later, we show how this assumption is not necessary. To handle
the third completely positive\ map in~\eqref{eq:Petz-3}, we again apply the
aforementioned result about sandwiching a Gaussian state by the square root of another.

\item After arriving at an explicit form for the Petz map by using
Ansatzes~1--3, we verify that this explicit form satisfies the equations in
\eqref{eq:Petz-equations} whenever the operators $A$ and $B$ are
Hilbert--Schmidt operators.

\item We finally employ a limiting argument to conclude that if
\eqref{eq:Petz-equations} is satisfied when $A$ and $B$ are Hilbert--Schmidt
operators, then the equations are satisfied when $A$ and $B$ are arbitrary
bounded operators. By a result of~\cite{petz86,petz88,PETZ}, we can finally
conclude that the Gaussian channel given in Theorem~\ref{thm:main} is the
unique quantum channel satisfying~\eqref{eq:Petz-equations}. This step then
concludes our proof of Theorem~\ref{thm:main}.
\end{enumerate}

In the subsections that follow, we give detailed proofs for each step above.

\subsection{Step~1: Sufficiency of focusing on zero-mean Gaussian states and
zero-displacement Gaussian channels}

\label{sec:first-step-zero-mean-etc}

As mentioned above, we employ Ansatz~1 in this first step, in which we work
with the explicit form of the Petz map in~\eqref{Petz}, in spite of the fact
that the inverse of a Gaussian density operator is unbounded. Let $\sigma$ be
a quantum Gaussian state with mean vector $s_{\sigma}$ and covariance matrix
$V_{\sigma}$, and let $\mathcal{N}$ be a quantum Gaussian channel with the
action on an input state as described in~\eqref{N cm}.

In this first step, we show how it suffices to consider the case $s_{\sigma
}=\delta=0$ in~\eqref{Petz}. To see this, consider the action of the Petz map
$\mathcal{P}_{\sigma,\mathcal{N}}$ on an arbitrary input state $\omega$:%
\begin{align}
\mathcal{P}_{\sigma,\mathcal{N}}(\omega)  &  =\sigma^{1/2}\mathcal{N}^{\dag
}\left(  \mathcal{N}(\sigma)^{-1/2}\omega\mathcal{N}(\sigma)^{-1/2}\right)
\sigma^{1/2}\,\\
&  =\left(  D_{s_{\sigma}}^{\dag}\sigma_{0}^{1/2}D_{s_{\sigma}}\right)
\mathcal{N}_{0}^{\dag}\left[  D_{\delta}D_{Xs_{\sigma}+\delta}^{\dag
}\mathcal{N}_{0}(\sigma_{0})^{-1/2}D_{Xs_{\sigma}+\delta}\,\omega
\,D_{Xs_{\sigma}+\delta}^{\dag}\mathcal{N}_{0}(\sigma_{0})^{-1/2}%
D_{Xs_{\sigma}+\delta}D_{\delta}^{\dag}\right] \nonumber\\
&  \qquad\times\left(  D_{s_{\sigma}}^{\dag}\sigma_{0}^{1/2}D_{s_{\sigma}%
}\right)  \,\\
&  =\left(  D_{s_{\sigma}}^{\dag}\sigma_{0}^{1/2}D_{s_{\sigma}}\right)
\mathcal{N}_{0}^{\dag}\left[  D_{Xs_{\sigma}}^{\dag}\mathcal{N}_{0}(\sigma
_{0})^{-1/2}D_{Xs_{\sigma}+\delta}\,\omega\,D_{Xs_{\sigma}+\delta}^{\dag
}\mathcal{N}_{0}(\sigma_{0})^{-1/2}D_{Xs_{\sigma}}\right] \nonumber\\
&  \qquad\times\left(  D_{s_{\sigma}}^{\dag}\sigma_{0}^{1/2}D_{s_{\sigma}%
}\right)  \,\\
&  =D_{s_{\sigma}}^{\dag}\sigma_{0}^{1/2}D_{s_{\sigma}}D_{X^{-1}(Xs_{\sigma}%
)}^{\dag}\mathcal{N}_{0}^{\dag}\left[  \mathcal{N}_{0}(\sigma_{0}%
)^{-1/2}D_{Xs_{\sigma}+\delta}\,\omega\,D_{Xs_{\sigma}+\delta}^{\dag
}\mathcal{N}_{0}(\sigma_{0})^{-1/2}\right] \nonumber\\
&  \qquad\times D_{X^{-1}(Xs_{\sigma})}D_{s_{\sigma}}^{\dag}\sigma_{0}%
^{1/2}D_{s_{\sigma}}\,\\
&  =D_{s_{\sigma}}^{\dag}\sigma_{0}^{1/2}\mathcal{N}_{0}^{\dag}\left[
\mathcal{N}_{0}(\sigma_{0})^{-1/2}D_{Xs_{\sigma}+\delta}\,\omega
\,D_{Xs_{\sigma}+\delta}^{\dag}\mathcal{N}_{0}(\sigma_{0})^{-1/2}\right]
\sigma_{0}^{1/2}D_{s_{\sigma}}\,\\
&  =D_{s_{\sigma}}^{\dag}\ \mathcal{P}_{\sigma_{0},\mathcal{N}_{0}}\left(
D_{Xs_{\sigma}+\delta}\,\omega\,D_{Xs_{\sigma}+\delta}^{\dag}\right)
\,D_{s_{\sigma}}\,. \label{Petz to Petz0}%
\end{align}
For the first equality, we use the definition of the Petz map and Ansatz~1.
The second equality follows from
\eqref{eq:state-means-out}--\eqref{eq:channel-state-means-out} and the fact
that $f(UAU^{\dag})=Uf(A)U^{\dag}$ for a function $f$, a unitary operator $U$,
and a Hermitian operator $A$. The third equality follows because $D_{\delta
}D_{Xs_{\sigma}+\delta}^{\dag}=D_{Xs_{\sigma}}^{\dag}e^{i\phi}$ for $\phi$ a
phase. The fourth equality follows from the adjoint channel covariance
relation in~\eqref{eq:adjoint-channel-cov-disps} and Ansatz~3. The fifth
equality follows because $D_{s_{\sigma}}D_{X^{-1}(Xs_{\sigma})}^{\dag
}=e^{i\varphi}I$ for some phase $\varphi$. The final equality follows by
recognizing the form of the Petz map $\mathcal{P}_{\sigma_{0},\mathcal{N}_{0}%
}$, corresponding to the zero-mean state $\sigma_{0}$ and the
zero-displacement channel $\mathcal{N}_{0}$.

The above reasoning suggests that we should focus on determining an explicit
form for $\mathcal{P}_{\sigma_{0},\mathcal{N}_{0}}(\omega)$. That is, the
above reasoning suggests that an arbitrary Petz map $\mathcal{P}%
_{\sigma,\mathcal{N}}$ can be realized as a serial concatenation of the
displacement $D_{Xs_{\sigma}+\delta}$, the Petz map $\mathcal{P}_{\sigma
_{0},\mathcal{N}_{0}}$, and the displacement $D_{s_{\sigma}}^{\dag}$. After we
give an explicit form for $\mathcal{P}_{\sigma_{0},\mathcal{N}_{0}}$ as a
quantum Gaussian channel with matrices $X_{P}$ and $Y_{P}$, it should become
clear why the displacement $\delta_{P}$\ in the Petz map $\mathcal{P}%
_{\sigma,\mathcal{N}}$ has the form in~\eqref{eq:Petz-displacement}.

\subsection{Step 2:\ Deducing a hypothesis for an explicit form for the Petz
map, by considering Gaussian input states}

In this step, we continue working with Ansatzes~1-3, with our main objective
being to arrive at a hypothesis for the action of the Petz recovery map
$\mathcal{P}_{\sigma_{0},\mathcal{N}_{0}}$ on the mean vector and covariance
matrix of an input Gaussian state. Here we consider the serial concatenation
of the three completely positive maps in~\eqref{eq:Petz-1}--\eqref{eq:Petz-3}.
We begin by considering the action of the last completely positive map on a
zero-mean Gaussian input state $\omega_{0}$. To this end, recall from~\cite{PS00} and 
\cite[Appendix~C]{Pirandola} that if $\omega_{0}$ and $\sigma_{0}$ are
zero-mean Gaussian states, then $\sqrt{\sigma_{0}}\omega_{0}\sqrt{\sigma_{0}}$
is an (unnormalized) Gaussian operator with zero mean vector and covariance
matrix given by%
\begin{equation}
V_{\sqrt{\sigma_{0}}\omega_{0}\sqrt{\sigma_{0}}}=V_{\sigma_{0}}-\left(
V_{\sqrt{\sigma_{0}}}-V_{\sigma_{0}}\right)  \left(  V_{\omega_{0}}%
+V_{\sigma_{0}}\right)  ^{-1}\left(  V_{\sqrt{\sigma_{0}}}-V_{\sigma_{0}%
}\right)  .\label{eq:sqrt-sandwich}%
\end{equation}
Applying a formula from~\cite[Lemma~2]{Kholevo1972} and \cite[Section~III]{PS00}, we find that%
\begin{equation}
V_{\sqrt{\sigma_{0}}}=\left(  \sqrt{I+\left(  V_{\sigma_{0}}\Omega\right)
^{-2}}+I\right)  V_{\sigma_{0}}\,,\label{CM sqrt G}%
\end{equation}
which is a symmetric matrix because $V_{\sigma_{0}}$ is. Indeed, consider that%
\begin{align}
V_{\sqrt{\sigma_{0}}}^{T}  & =\left[  \left(  \sqrt{I+\left(  V_{\sigma_{0}%
}\Omega\right)  ^{-2}}+I\right)  V_{\sigma_{0}}\right]  ^{T}=V_{\sigma_{0}%
}\left(  \sqrt{I+\left(  \Omega V_{\sigma_{0}}\right)  ^{-2}}+I\right)  \\
& =\Omega^{-1}\Omega V_{\sigma_{0}}\left(  \sqrt{I+\left(  \Omega
V_{\sigma_{0}}\right)  ^{-2}}+I\right)  =\Omega^{-1}\left(  \sqrt{I+\left(
\Omega V_{\sigma_{0}}\right)  ^{-2}}+I\right)  \Omega V_{\sigma_{0}}\\
& =\left(  \sqrt{\Omega^{-1}\left[  I+\left(  \Omega V_{\sigma_{0}}\right)
^{-2}\right]  \Omega}+I\right)  V_{\sigma_{0}}=\left(  \sqrt{\left[  I+\left(
\Omega^{-1}\Omega V_{\sigma_{0}}\Omega\right)  ^{-2}\right]  }+I\right)
V_{\sigma_{0}}\\
& =\left(  \sqrt{I+\left(  V_{\sigma_{0}}\Omega\right)  ^{-2}}+I\right)
V_{\sigma_{0}}=V_{\sqrt{\sigma_{0}}}.
\end{align}
The equality in~\eqref{CM sqrt G}\ implies that%
\begin{equation}
V_{\sqrt{\sigma_{0}}}-V_{\sigma_{0}}=\sqrt{I+\left(  V_{\sigma_{0}}%
\Omega\right)  ^{-2}}V_{\sigma_{0}},
\end{equation}
and in turn, after substituting into~\eqref{eq:sqrt-sandwich}, that%
\begin{equation}
V_{\sqrt{\sigma_{0}}\omega_{0}\sqrt{\sigma_{0}}}=V_{\sigma_{0}}-\sqrt
{I+\left(  V_{\sigma_{0}}\Omega\right)  ^{-2}}V_{\sigma_{0}}\left(
V_{\omega_{0}}+V_{\sigma_{0}}\right)  ^{-1}V_{\sigma_{0}}\sqrt{I+\left(
\Omega V_{\sigma_{0}}\right)  ^{-2}}.\label{cm sqrt sandwiched}%
\end{equation}
Thus,~\eqref{cm sqrt sandwiched} establishes the action of the completely
positive map $(\cdot)\rightarrow\sqrt{\sigma_{0}}(\cdot)\sqrt{\sigma_{0}}$ on
an arbitrary zero-mean Gaussian state $\omega_{0}$.

From this discussion we already start seeing that the Petz map constructed out
of a Gaussian state $\sigma$\ and a Gaussian channel $\mathcal{N}$\ should
send normalized Gaussian states to normalized Gaussian states, because (i)
conjugation by the square root of a Gaussian state (or the inverse square root
of a Gaussian state as we will see) preserves the Gaussian form; (ii) the
adjoint of a Gaussian channel is still Gaussian; and (iii) the Petz map is a
priori known to be trace-preserving whenever $\mathcal{N}(\sigma)$ is a
faithful state~\cite{petz86,petz88,PETZ}. Then,~\cite[Theorem~III.1]{dePalma}
ensures that $\mathcal{P}$ must act as in~\eqref{N cm}, for some $X_{P}$,
$Y_{P}$, and $\delta_{P}$ to be determined.

With this preliminary identity in hand, we are ready to determine a hypothesis
for the explicit action of $\mathcal{P}_{\sigma_{0},\mathcal{N}_{0}}$.
For the sake of simplicity, we consider the input Gaussian state to have
vanishing first moments. In any case, since we are working to deduce a
hypothesis for an explicit form for the Petz map, this is by no means a loss
of generality. By applying~\eqref{cm sqrt sandwiched}\ and\ Ansatz~2 (that the
following density operator transformation $\omega\rightarrow\omega^{-1}$
induces the transformation $V_{\omega}\rightarrow-V_{\omega}$ on the level of
covariance matrices), we can conclude that the completely positive map in
\eqref{eq:Petz-1} has the following effect on covariance matrices:%
\begin{multline}
V_{\sqrt{\mathcal{N}_{0}(\sigma_{0})}^{-1}\omega_{0}\sqrt{\mathcal{N}%
_{0}(\sigma_{0})}^{-1}}\\
=-V_{\mathcal{N}(\sigma)}-\sqrt{I+\left(  V_{\mathcal{N}(\sigma)}%
\Omega\right)  ^{-2}}V_{\mathcal{N}(\sigma)}\left(  V_{\omega}-V_{\mathcal{N}%
(\sigma)}\right)  ^{-1}V_{\mathcal{N}(\sigma)}\sqrt{I+\left(  \Omega
V_{\mathcal{N}(\sigma)}\right)  ^{-2}}.
\end{multline}
In the above, we have also used the identities $V_{\mathcal{N}_{0}(\sigma
_{0})}=V_{\mathcal{N}(\sigma)}$ and $V_{\omega_{0}}=V_{\omega}$. So now we
consider further concatenating with the completely positive map in
\eqref{eq:Petz-2}, by applying~\eqref{N dag cm} and Ansatz~3 (that $X$ is
invertible):%
\begin{multline}
V_{\mathcal{N}_{0}^{\dag}(\sqrt{\mathcal{N}_{0}(\sigma_{0})}^{-1}\omega
_{0}\sqrt{\mathcal{N}_{0}(\sigma_{0})}^{-1})}%
=\label{eq:adjoint-and-inverse-sqrt}\\
X^{-1}\left[  -V_{\mathcal{N}(\sigma)}-\sqrt{I+\left(  V_{\mathcal{N}(\sigma
)}\Omega\right)  ^{-2}}V_{\mathcal{N}(\sigma)}\left(  V_{\omega}%
-V_{\mathcal{N}(\sigma)}\right)  ^{-1}V_{\mathcal{N}(\sigma)}\sqrt{I+\left(
\Omega V_{\mathcal{N}(\sigma)}\right)  ^{-2}}+Y\right]  X^{-T}.
\end{multline}
But consider that $V_{\mathcal{N}(\sigma)}=XV_{\sigma}X^{T}+Y$, so that
\eqref{eq:adjoint-and-inverse-sqrt} simplifies as follows:%
\begin{align}
&  V_{\mathcal{N}_{0}^{\dag}(\sqrt{\mathcal{N}_{0}(\sigma_{0})}^{-1}\omega
_{0}\sqrt{\mathcal{N}_{0}(\sigma_{0})}^{-1})}\nonumber\\
&  =X^{-1}\Bigg[-\left(  XV_{\sigma}X^{T}+Y\right)  -\sqrt{I+\left(
V_{\mathcal{N}(\sigma)}\Omega\right)  ^{-2}}V_{\mathcal{N}(\sigma)}\left(
V_{\omega}-V_{\mathcal{N}(\sigma)}\right)  ^{-1}V_{\mathcal{N}(\sigma)}%
\sqrt{I+\left(  \Omega V_{\mathcal{N}(\sigma)}\right)  ^{-2}}+Y\Bigg]X^{-T}%
\nonumber\\
&  =X^{-1}\left[  -XV_{\sigma}X^{T}-\sqrt{I+\left(  V_{\mathcal{N}(\sigma
)}\Omega\right)  ^{-2}}V_{\mathcal{N}(\sigma)}\left(  V_{\omega}%
-V_{\mathcal{N}(\sigma)}\right)  ^{-1}V_{\mathcal{N}(\sigma)}\sqrt{I+\left(
\Omega V_{\mathcal{N}(\sigma)}\right)  ^{-2}}\right]  X^{-T}\\
&  =-V_{\sigma}-X^{-1}\sqrt{I+\left(  V_{\mathcal{N}(\sigma)}\Omega\right)
^{-2}}V_{\mathcal{N}(\sigma)}\left(  V_{\omega}-V_{\mathcal{N}(\sigma
)}\right)  ^{-1}V_{\mathcal{N}(\sigma)}\sqrt{I+\left(  \Omega V_{\mathcal{N}%
(\sigma)}\right)  ^{-2}}X^{-T}.
\end{align}
So then we can finally consider the serial concatenation of the three
completely positive maps in~\eqref{eq:Petz-1}--\eqref{eq:Petz-3}:%
\begin{align}
&  V_{\sqrt{\sigma_{0}}\mathcal{N}_{0}^{\dag}(\sqrt{\mathcal{N}_{0}(\sigma
_{0})}^{-1}\omega_{0}\sqrt{\mathcal{N}_{0}(\sigma_{0})}^{-1})\sqrt{\sigma_{0}%
}}\nonumber\\
&  =V_{\sigma}-\sqrt{I+\left(  V_{\sigma}\Omega\right)  ^{-2}}V_{\sigma
}\nonumber\\
&  \qquad\times\left(  -V_{\sigma}-X^{-1}\sqrt{I+\left(  V_{\mathcal{N}%
(\sigma)}\Omega\right)  ^{-2}}V_{\mathcal{N}(\sigma)}\left(  V_{\omega
}-V_{\mathcal{N}(\sigma)}\right)  ^{-1}V_{\mathcal{N}(\sigma)}\sqrt{I+\left(
\Omega V_{\mathcal{N}(\sigma)}\right)  ^{-2}}X^{-T}+V_{\sigma}\right)
^{-1}\nonumber\\
&  \qquad\times V_{\sigma}\sqrt{I+\left(  \Omega V_{\sigma}\right)  ^{-2}}%
\end{align}%
\begin{align}
&  =V_{\sigma}-\sqrt{I+\left(  V_{\sigma}\Omega\right)  ^{-2}}V_{\sigma
}\nonumber\\
&  \qquad\times\left(  -X^{-1}\sqrt{I+\left(  V_{\mathcal{N}(\sigma)}%
\Omega\right)  ^{-2}}V_{\mathcal{N}(\sigma)}\left(  V_{\omega}-V_{\mathcal{N}%
(\sigma)}\right)  ^{-1}V_{\mathcal{N}(\sigma)}\sqrt{I+\left(  \Omega
V_{\mathcal{N}(\sigma)}\right)  ^{-2}}X^{-T}\right)  ^{-1}\nonumber\\
&  \qquad\times V_{\sigma}\sqrt{I+\left(  \Omega V_{\sigma}\right)  ^{-2}}\\
&  =V_{\sigma}+\sqrt{I+\left(  V_{\sigma}\Omega\right)  ^{-2}}V_{\sigma}%
X^{T}\sqrt{I+\left(  \Omega V_{\mathcal{N}(\sigma)}\right)  ^{-2}}%
^{-1}V_{\mathcal{N}(\sigma)}^{-1}\left(  V_{\omega}-V_{\mathcal{N}(\sigma
)}\right)  \nonumber\\
&  \qquad\times V_{\mathcal{N}(\sigma)}^{-1}\sqrt{I+\left(  V_{\mathcal{N}%
(\sigma)}\Omega\right)  ^{-2}}^{-1}XV_{\sigma}\sqrt{I+\left(  V_{\sigma}%
\Omega\right)  ^{-2}}.\label{eq:last-step-explicit-petz}%
\end{align}

An inspection of~\eqref{eq:last-step-explicit-petz}\ above suggests that the
Petz map $\mathcal{P}_{\sigma_{0},\mathcal{N}_{0}}$\ is a quantum Gaussian
channel with the following action on an input covariance matrix $V_{\omega}$:%
\begin{equation}
V_{\mathcal{P}_{\sigma_{0},\mathcal{N}_{0}}(\omega_{0})}=X_{P}V_{\omega}%
X_{P}^{T}+Y_{P},\label{Petz w_0}%
\end{equation}
where%
\begin{align}
X_{P} &  \equiv\sqrt{I+\left(  V_{\sigma}\Omega\right)  ^{-2}}V_{\sigma}%
X^{T}\sqrt{I+\left(  \Omega V_{\mathcal{N}(\sigma)}\right)  ^{-2}}%
^{-1}V_{\mathcal{N}(\sigma)}^{-1},\label{XP}\\
Y_{P} &  \equiv V_{\sigma}-X_{P}V_{\mathcal{N}(\sigma)}X_{P}^{T}.\label{YP}%
\end{align}
Combining with the development in Section~\ref{sec:first-step-zero-mean-etc},
the results in~\eqref{Petz w_0},~\eqref{Petz to Petz0} and
\cite[Theorem~III.1]{dePalma} imply that in general
\begin{equation}
\mathcal{P}_{\sigma,\mathcal{N}}:\ \left\{
\begin{array}
[c]{lcl}%
V & \longmapsto & X_{P}VX_{P}^{T}+Y_{P}\\
s & \longmapsto & X_{P}s+\delta_{P}%
\end{array}
\right.  ,\label{Petz cm}%
\end{equation}
where%
\begin{equation}
\delta_{P}\equiv s_{\sigma}-X_{P}\left(  Xs_{\sigma}+\delta\right)
\,,\label{delta_P}%
\end{equation}
and $\delta$ is the vector appearing in~\eqref{N cm}; it follows because%
\begin{equation}
\mathcal{P}_{\sigma,\mathcal{N}}(\omega)=D_{s_{\sigma}}^{\dag}\mathcal{P}%
_{\sigma_{0},\mathcal{N}_{0}}\left(  D_{Xs_{\sigma}+\delta}\omega
D_{Xs_{\sigma}+\delta}^{\dag}\right)  D_{s_{\sigma}},
\end{equation}
which implies that%
\begin{equation}
s_{\mathcal{P}_{\sigma,\mathcal{N}}(\omega)}=X_{P}(s_{\omega}-Xs_{\sigma
}-\delta)+s_{\sigma}.
\end{equation}
So by using Ansatzes~1-3, we have arrived at our hypothesis
\eqref{Petz cm}\ for the Gaussian form of the Petz map $\mathcal{P}%
_{\sigma,\mathcal{N}}$. In the next two sections, we give a detailed proof
that the Gaussian channel specified in~\eqref{Petz cm}\ is indeed equal to the
Petz map $\mathcal{P}_{\sigma,\mathcal{N}}$.

\subsection{Step 3: The Gaussian Petz map satisfies the Petz equations for all
Hilbert--Schmidt operators}

\label{subsec Hilbert-Schmidt}

In this section, we prove that the hypothesis~\eqref{Petz cm} for the Petz map
satisfies the equations in~\eqref{eq:Petz-equations}\ for all Hilbert--Schmidt
operators. Recall that an operator $T$ is Hilbert--Schmidt if%
\begin{equation}
\left\Vert T\right\Vert _{2}\equiv\sqrt{\operatorname{Tr}[T^{\dag}T]}<\infty.
\end{equation}
It can be shown that the Hilbert-Schmidt operators defined on a given Hilbert space form a Hilbert space themselves, once equipped with the product $\langle T_{1}, T_{2}\rangle\equiv \operatorname{Tr} [T_{1}^{\dag} T_{2}]$
\cite{HOLEVO}.
Let $T$ act on a tensor product of $n$ separable Hilbert spaces (i.e., $n$
modes). Its characteristic function is defined by%
\begin{equation}
\chi_{T}(w)=\operatorname{Tr}[TD_{-w}],
\end{equation}
where $w\in\mathbb{R}^{2n}$. Thus, we can write $T$ in terms of its
characteristic function as%
\begin{equation}
T=\int\frac{d^{2n}w}{(2\pi)^{n}}\ \chi_{T}(w)\ D_{w}.
\end{equation}
In fact, the above one-to-one mapping between operators and characteristic functions can be viewed as an isometry between two a priori very different Hilbert spaces, namely that formed by all Hilbert-Schmidt operators on $n$ modes, and that formed by all complex-valued, square-integrable functions $\mathds{R}^{2n}\rightarrow \mathds{C}$, customarily denoted by $\mathcal{L}^{2}\left(\mathds{R}^{2n}\right)$ \cite[Theorem 5.3.3]{HOLEVO}.

Suppose that $T_{1}$ and $T_{2}$ are Hilbert--Schmidt operators. In order to
demonstrate that our hypothesis~\eqref{Petz cm} for $\mathcal{P}%
_{\sigma,\mathcal{N}}$ is in fact correct, we first show that the following
equation is satisfied for this choice and for all Hilbert--Schmidt $T_{1}$ and
$T_{2}$:%
\begin{equation}
\langle T_{2},\mathcal{N}^{\dag}(T_{1})\rangle_{\sigma}=\langle\mathcal{P}%
_{\sigma,\mathcal{N}}^{\dag}(T_{2}),T_{1}\rangle_{\mathcal{N}(\sigma
)}.\label{eq:Petz-HS-ops}%
\end{equation}
Using definitions and an expansion of $T_{1}$ and $T_{2}$ in terms of their
characteristic functions $\chi_{T_{1}}(w_{1})$ and $\chi_{T_{2}}(w_{2})$,
respectively, where $w_{1},w_{2}\in\mathbb{R}^{2n}$, we find that
\eqref{eq:Petz-HS-ops} is equivalent to%
\begin{multline}
\int\int\frac{d^{2n}w_{1}\ d^{2n}w_{2}}{(2\pi)^{2n}}\chi_{T_{2}}^{\ast}%
(w_{2})\chi_{T_{1}}(w_{1})\operatorname{Tr}[\sigma^{1/2}D_{-w_{2}}\sigma
^{1/2}\mathcal{N}^{\dag}(D_{w_{1}})]\\
=\int\int\frac{d^{2n}w_{1}\ d^{2n}w_{2}}{(2\pi)^{2n}}\chi_{T_{2}}^{\ast}%
(w_{2})\chi_{T_{1}}(w_{1})\operatorname{Tr}[\mathcal{P}_{\sigma,\mathcal{N}%
}^{\dag}(D_{-w_{2}})\mathcal{N}(\sigma)^{1/2}D_{w_{1}}\mathcal{N}%
(\sigma)^{1/2}].
\end{multline}
Thus, if we show that the following holds for all $w_{1},w_{2}\in
\mathbb{R}^{2n}$%
\begin{equation}
\operatorname{Tr}[\sigma^{1/2}D_{-w_{2}}\sigma^{1/2}\mathcal{N}^{\dag
}(D_{w_{1}})]=\operatorname{Tr}[\mathcal{P}_{\sigma,\mathcal{N}}^{\dag
}(D_{-w_{2}})\mathcal{N}(\sigma)^{1/2}D_{w_{1}}\mathcal{N}(\sigma
)^{1/2}],\label{eq:displacement-relation-Petz-eq}%
\end{equation}
then the statement in~\eqref{eq:Petz-HS-ops}\ is shown for all
Hilbert--Schmidt operators. So we proceed with proving~\eqref{eq:displacement-relation-Petz-eq}.

We first show that it suffices to verify
\eqref{eq:displacement-relation-Petz-eq} when $\sigma$ is a zero-mean Gaussian
state and $\mathcal{N}$ is a zero-displacement Gaussian channel. Here we make
use of~\eqref{eq:state-means-out},~\eqref{eq:channel-means-out}, and
\eqref{eq:channel-state-means-out}. Consider that%
\begin{align}
\sigma^{1/2}  &  =D_{s_{\sigma}}^{\dag}\sigma_{0}^{1/2}D_{s_{\sigma}},\\
\mathcal{N}(\sigma)  &  =D_{Xs_{\sigma}+\delta}^{\dag}\mathcal{N}_{0}%
(\sigma_{0})D_{Xs_{\sigma}+\delta},\\
\mathcal{N}(\cdot)  &  =D_{\delta}^{\dag}\mathcal{N}_{0}(\cdot)D_{\delta},\\
\mathcal{P}_{\sigma,\mathcal{N}}(\cdot)  &  =D_{\delta_{P}}^{\dag}%
\mathcal{P}_{\sigma_{0},\mathcal{N}_{0}}(\cdot)D_{\delta_{P}},
\end{align}
where $\delta_{P}$ is defined as in~\eqref{eq:Petz-displacement}. We can then
rewrite the left-hand side of~\eqref{eq:displacement-relation-Petz-eq} as%
\begin{align}
&  \operatorname{Tr}[\sigma^{1/2}D_{-w_{2}}\sigma^{1/2}\mathcal{N}^{\dag
}(D_{w_{1}})]\nonumber\\
&  =\operatorname{Tr}[\mathcal{N}(\sigma^{1/2}D_{-w_{2}}\sigma^{1/2})D_{w_{1}%
}]\\
&  =\operatorname{Tr}[D_{\delta}^{\dag}\mathcal{N}_{0}(D_{s_{\sigma}}^{\dag
}\sigma_{0}^{1/2}D_{s_{\sigma}}D_{-w_{2}}D_{s_{\sigma}}^{\dag}\sigma_{0}%
^{1/2}D_{s_{\sigma}})D_{\delta}D_{w_{1}}]\\
&  =\operatorname{Tr}[D_{Xs_{\sigma}+\delta}^{\dag}\mathcal{N}_{0}(\sigma
_{0}^{1/2}D_{s_{\sigma}}D_{-w_{2}}D_{s_{\sigma}}^{\dag}\sigma_{0}%
^{1/2})D_{Xs_{\sigma}+\delta}D_{w_{1}}]\\
&  =\operatorname{Tr}[\mathcal{N}_{0}(\sigma_{0}^{1/2}D_{s_{\sigma}}D_{-w_{2}%
}D_{s_{\sigma}}^{\dag}\sigma_{0}^{1/2})D_{Xs_{\sigma}+\delta}D_{w_{1}%
}D_{Xs_{\sigma}+\delta}^{\dag}]\\
&  =\exp(-i\left(  Xs_{\sigma}+\delta\right)  ^{T}\Omega w_{1}+is_{\sigma}%
^{T}\Omega w_{2})\operatorname{Tr}[\mathcal{N}_{0}(\sigma_{0}^{1/2}D_{-w_{2}%
}\sigma_{0}^{1/2})D_{w_{1}}]\\
&  =\exp(-i\left(  Xs_{\sigma}+\delta\right)  ^{T}\Omega w_{1}+is_{\sigma}%
^{T}\Omega w_{2})\operatorname{Tr}[\sigma_{0}^{1/2}D_{-w_{2}}\sigma_{0}%
^{1/2}\mathcal{N}_{0}^{\dag}(D_{w_{1}})].
\label{eq:verify-Petz-LHS-zero-means}%
\end{align}
We can rewrite the right-hand side of~\eqref{eq:displacement-relation-Petz-eq}
as%
\begin{align}
&  \operatorname{Tr}[\mathcal{P}_{\sigma,\mathcal{N}}^{\dag}(D_{-w_{2}%
})\mathcal{N}(\sigma)^{1/2}D_{w_{1}}\mathcal{N}(\sigma)^{1/2}]\nonumber\\
&  =\operatorname{Tr}[D_{-w_{2}}\mathcal{P}_{\sigma,\mathcal{N}}%
(\mathcal{N}(\sigma)^{1/2}D_{w_{1}}\mathcal{N}(\sigma)^{1/2})]\\
&  =\operatorname{Tr}[D_{-w_{2}}D_{\delta_{P}}^{\dag}\mathcal{P}_{\sigma
_{0},\mathcal{N}_{0}}(D_{Xs_{\sigma}+\delta}^{\dag}\mathcal{N}_{0}(\sigma
_{0})^{1/2}D_{Xs_{\sigma}+\delta}D_{w_{1}}D_{Xs_{\sigma}+\delta}^{\dag
}\mathcal{N}_{0}(\sigma_{0})^{1/2}D_{Xs_{\sigma}+\delta})D_{\delta_{P}}]\\
&  =\operatorname{Tr}[D_{-w_{2}}D_{\delta_{P}}^{\dag}D_{X_{P}\left[
Xs_{\sigma}+\delta\right]  }^{\dag}\mathcal{P}_{\sigma_{0},\mathcal{N}_{0}%
}(\mathcal{N}_{0}(\sigma_{0})^{1/2}D_{Xs_{\sigma}+\delta}D_{w_{1}%
}D_{Xs_{\sigma}+\delta}^{\dag}\mathcal{N}_{0}(\sigma_{0})^{1/2})D_{X_{P}%
\left[  Xs_{\sigma}+\delta\right]  }D_{\delta_{P}}].
\label{eq:Petz-with-means-RHS-last-line}%
\end{align}
Considering that%
\begin{equation}
D_{X_{P}\left[  Xs_{\sigma}+\delta\right]  }D_{\delta_{P}}=D_{\delta_{P}%
+X_{P}\left[  Xs_{\sigma}+\delta\right]  }e^{i\phi}=D_{s_{\sigma}}e^{i\phi},
\end{equation}
which follows from~\eqref{eq:Petz-displacement}\ and
\eqref{comp displacement}, we find that
\eqref{eq:Petz-with-means-RHS-last-line} is equal to%
\begin{align}
&  \operatorname{Tr}[D_{-w_{2}}D_{s_{\sigma}}^{\dag}\mathcal{P}_{\sigma
_{0},\mathcal{N}_{0}}(\mathcal{N}_{0}(\sigma_{0})^{1/2}D_{Xs_{\sigma}+\delta
}D_{w_{1}}D_{Xs_{\sigma}+\delta}^{\dag}\mathcal{N}_{0}(\sigma_{0}%
)^{1/2})D_{s_{\sigma}}]\nonumber\\
&  =\operatorname{Tr}[D_{s_{\sigma}}D_{-w_{2}}D_{s_{\sigma}}^{\dag}%
\mathcal{P}_{\sigma_{0},\mathcal{N}_{0}}(\mathcal{N}_{0}(\sigma_{0}%
)^{1/2}D_{Xs_{\sigma}+\delta}D_{w_{1}}D_{Xs_{\sigma}+\delta}^{\dag}%
\mathcal{N}_{0}(\sigma_{0})^{1/2})]\\
&  =\exp(-i\left(  Xs_{\sigma}+\delta\right)  ^{T}\Omega w_{1}+is_{\sigma}%
^{T}\Omega w_{2})\operatorname{Tr}[D_{-w_{2}}\mathcal{P}_{\sigma
_{0},\mathcal{N}_{0}}(\mathcal{N}_{0}(\sigma_{0})^{1/2}D_{w_{1}}%
\mathcal{N}_{0}(\sigma_{0})^{1/2})]\\
&  =\exp(-i\left(  Xs_{\sigma}+\delta\right)  ^{T}\Omega w_{1}+is_{\sigma}%
^{T}\Omega w_{2})\operatorname{Tr}[\mathcal{P}_{\sigma_{0},\mathcal{N}_{0}%
}^{\dag}(D_{-w_{2}})\mathcal{N}_{0}(\sigma_{0})^{1/2}D_{w_{1}}\mathcal{N}%
_{0}(\sigma_{0})^{1/2}]. \label{eq:verify-Petz-RHS-zero-means}%
\end{align}
Observe that the phases in~\eqref{eq:verify-Petz-LHS-zero-means} and
\eqref{eq:verify-Petz-RHS-zero-means} are equal. Thus, if the goal is to show
the equality in~\eqref{eq:displacement-relation-Petz-eq}, then our above
development proves that it suffices to establish the following equality:%
\begin{equation}
\operatorname{Tr}[\sigma_{0}^{1/2}D_{-w_{2}}\sigma_{0}^{1/2}\mathcal{N}%
_{0}^{\dag}(D_{w_{1}})]=\operatorname{Tr}[\mathcal{P}_{\sigma_{0}%
,\mathcal{N}_{0}}^{\dag}(D_{-w_{2}})\mathcal{N}_{0}(\sigma_{0})^{1/2}D_{w_{1}%
}\mathcal{N}_{0}(\sigma_{0})^{1/2}]. \label{eq:zero-mean-Petz-eqs}%
\end{equation}
So now we focus on establishing~\eqref{eq:zero-mean-Petz-eqs}.

To begin with, consider from~\eqref{eq:adjoint-on-disps}\ that%
\begin{equation}
\mathcal{N}_{0}^{\dag}(D_{w_{1}})=D_{\Omega X^{T}\Omega^{T}w_{1}}\exp\!\left(
-\frac{1}{4}(\Omega^{T}w_{1})^{T}Y\Omega^{T}w_{1}\right)  .
\end{equation}
Thus, the left-hand side of~\eqref{eq:zero-mean-Petz-eqs} reduces to%
\begin{equation}
\operatorname{Tr}[\sigma_{0}^{1/2}D_{-w_{2}}\sigma_{0}^{1/2}D_{\Omega
X^{T}\Omega^{T}w_{1}}]\exp\!\left(  -\frac{1}{4}(\Omega^{T}w_{1})^{T}%
Y\Omega^{T}w_{1}\right)  . \label{eq:LHS-final-petz}%
\end{equation}
Similarly, from~\eqref{eq:adjoint-on-disps} and~\eqref{Petz cm}, we have that%
\begin{equation}
\mathcal{P}_{\sigma_{0},\mathcal{N}_{0}}^{\dag}(D_{-w_{2}})=D_{-\Omega
X_{P}^{T}\Omega^{T}w_{2}}\exp\!\left(  -\frac{1}{4}(\Omega^{T}w_{2})^{T}%
Y_{P}\Omega^{T}w_{2}\right)  ,
\end{equation}
so that the right-hand side of~\eqref{eq:zero-mean-Petz-eqs} reduces to%
\begin{equation}
\operatorname{Tr}[D_{-\Omega X_{P}^{T}\Omega^{T}w_{2}}\mathcal{N}_{0}%
(\sigma_{0})^{1/2}D_{w_{1}}\mathcal{N}_{0}(\sigma_{0})^{1/2}]\exp\!\left(
-\frac{1}{4}(\Omega^{T}w_{2})^{T}Y_{P}\Omega^{T}w_{2}\right)  .
\label{eq:RHS-final-petz}%
\end{equation}
So we should show the equality of~\eqref{eq:LHS-final-petz} and
\eqref{eq:RHS-final-petz}, in order to establish the equality in~\eqref{eq:zero-mean-Petz-eqs}.

To this end, Lemma~\ref{lemma sqrt sandwich} below is helpful for us. Invoking
it, we find that the left-most expression in~\eqref{eq:LHS-final-petz} reduces
as%
\begin{align}
&  \operatorname{Tr}[\sigma_{0}^{1/2}D_{-w_{2}}\sigma_{0}^{1/2}D_{\Omega
X^{T}\Omega^{T}w_{1}}]\nonumber\\
&  =\exp\!\Bigg(-\frac{1}{4}\left(  \Omega X^{T}\Omega^{T}w_{1}\right)
^{T}\Omega^{T}V_{\sigma}\Omega\Omega X^{T}\Omega^{T}w_{1}-\frac{1}{4}w_{2}%
^{T}\Omega^{T}V_{\sigma}\Omega w_{2}\nonumber\\
&  \qquad\qquad\qquad+\frac{1}{2}\left(  \Omega X^{T}\Omega^{T}w_{1}\right)
^{T}\Omega^{T}\sqrt{I+(V_{\sigma}\Omega)^{-2}}V_{\sigma}\Omega w_{2}\Bigg)\\
&  =\exp\!\left(  -\frac{1}{4}\left(  \Omega^{T}w_{1}\right)  ^{T}XV_{\sigma
}X^{T}\Omega^{T}w_{1}-\frac{1}{4}w_{2}^{T}\Omega^{T}V_{\sigma}\Omega
w_{2}-\frac{1}{2}\left(  \Omega^{T}w_{1}\right)  ^{T}X\sqrt{I+(V_{\sigma
}\Omega)^{-2}}V_{\sigma}\Omega w_{2}\right)  .
\end{align}
So this implies that~\eqref{eq:LHS-final-petz} is equal to%
\begin{multline}
\exp\!\Bigg(-\frac{1}{4}\left(  \Omega^{T}w_{1}\right)  ^{T}XV_{\sigma}%
X^{T}\Omega^{T}w_{1}-\frac{1}{4}w_{2}^{T}\Omega^{T}V_{\sigma}\Omega
w_{2}\label{eq:final-final-LHS-Petz}\\
-\frac{1}{2}\left(  \Omega^{T}w_{1}\right)  ^{T}X\sqrt{I+(V_{\sigma}%
\Omega)^{-2}}V_{\sigma}\Omega w_{2}-\frac{1}{4}(\Omega^{T}w_{1})^{T}%
Y\Omega^{T}w_{1}\Bigg)\\
=\exp\!\left(  -\frac{1}{4}\left(  \Omega^{T}w_{1}\right)  ^{T}V_{\mathcal{N}%
(\sigma)}\Omega^{T}w_{1}-\frac{1}{4}w_{2}^{T}\Omega^{T}V_{\sigma}\Omega
w_{2}-\frac{1}{2}\left(  \Omega^{T}w_{1}\right)  ^{T}X\sqrt{I+(V_{\sigma
}\Omega)^{-2}}V_{\sigma}\Omega w_{2}\right)  .
\end{multline}
Invoking Lemma~\ref{lemma sqrt sandwich} again, we find that the left-most
expression in~\eqref{eq:RHS-final-petz} reduces as%
\begin{align}
&  \operatorname{Tr}[D_{-\Omega X_{P}^{T}\Omega^{T}w_{2}}\mathcal{N}%
_{0}(\sigma_{0})^{1/2}D_{w_{1}}\mathcal{N}_{0}(\sigma_{0})^{1/2}]\nonumber\\
&  =\exp\!\Bigg(-\frac{1}{4}w_{1}^{T}\Omega^{T}V_{\mathcal{N}(\sigma)}\Omega
w_{1}-\frac{1}{4}\left(  \Omega X_{P}^{T}\Omega^{T}w_{2}\right)  ^{T}%
\Omega^{T}V_{\mathcal{N}(\sigma)}\Omega\Omega X_{P}^{T}\Omega^{T}%
w_{2}\nonumber\\
&  \qquad\qquad+\frac{1}{2}w_{1}^{T}\Omega^{T}\sqrt{I+(V_{\mathcal{N}(\sigma
)}\Omega)^{-2}}V_{\mathcal{N}(\sigma)}\Omega\Omega X_{P}^{T}\Omega^{T}%
w_{2}\Bigg)\\
&  =\exp\!\Bigg(-\frac{1}{4}w_{1}^{T}\Omega^{T}V_{\mathcal{N}(\sigma)}\Omega
w_{1}-\frac{1}{4}\left(  \Omega^{T}w_{2}\right)  ^{T}X_{P}V_{\mathcal{N}%
(\sigma)}X_{P}^{T}\Omega^{T}w_{2}\nonumber\\
&  \qquad\qquad-\frac{1}{2}w_{1}^{T}\Omega^{T}\sqrt{I+(V_{\mathcal{N}(\sigma
)}\Omega)^{-2}}V_{\mathcal{N}(\sigma)}X_{P}^{T}\Omega^{T}w_{2}\Bigg).
\end{align}
So this implies that~\eqref{eq:RHS-final-petz} is equal to%
\begin{multline}
\exp\!\Bigg(-\frac{1}{4}w_{1}^{T}\Omega^{T}V_{\mathcal{N}(\sigma)}\Omega
w_{1}-\frac{1}{4}\left(  \Omega^{T}w_{2}\right)  ^{T}X_{P}V_{\mathcal{N}%
(\sigma)}X_{P}^{T}\Omega^{T}w_{2}\\
\qquad-\frac{1}{2}w_{1}^{T}\Omega^{T}\sqrt{I+(V_{\mathcal{N}(\sigma)}%
\Omega)^{-2}}V_{\mathcal{N}(\sigma)}X_{P}^{T}\Omega^{T}w_{2}-\frac{1}%
{4}(\Omega^{T}w_{2})^{T}Y_{P}\Omega^{T}w_{2}\Bigg)\\
=\exp\!\left(  -\frac{1}{4}w_{1}^{T}\Omega^{T}V_{\mathcal{N}(\sigma)}\Omega
w_{1}-\frac{1}{4}\left(  \Omega^{T}w_{2}\right)  ^{T}V_{\sigma}\Omega^{T}%
w_{2}-\frac{1}{2}w_{1}^{T}\Omega^{T}\sqrt{I+(V_{\mathcal{N}(\sigma)}%
\Omega)^{-2}}V_{\mathcal{N}(\sigma)}X_{P}^{T}\Omega^{T}w_{2}\right)  .
\label{eq:final-almost-there}%
\end{multline}
Consider that%
\begin{align}
&  \sqrt{I+(V_{\mathcal{N}(\sigma)}\Omega)^{-2}}V_{\mathcal{N}(\sigma)}%
X_{P}^{T}\nonumber\\
&  =\sqrt{I+(V_{\mathcal{N}(\sigma)}\Omega)^{-2}}V_{\mathcal{N}(\sigma
)}\left(  \sqrt{I+\left(  V_{\sigma}\Omega\right)  ^{-2}}V_{\sigma}X^{T}%
\sqrt{I+\left(  \Omega V_{\mathcal{N}(\sigma)}\right)  ^{-2}}^{-1}%
V_{\mathcal{N}(\sigma)}^{-1}\right)  ^{T}\\
&  =\sqrt{I+(V_{\mathcal{N}(\sigma)}\Omega)^{-2}}V_{\mathcal{N}(\sigma
)}V_{\mathcal{N}(\sigma)}^{-1}\sqrt{I+\left(  V_{\mathcal{N}(\sigma)}%
\Omega\right)  ^{-2}}^{-1}X\sqrt{I+\left(  V_{\sigma}\Omega\right)  ^{-2}%
}V_{\sigma}\\
&  =X\sqrt{I+\left(  V_{\sigma}\Omega\right)  ^{-2}}V_{\sigma},
\end{align}
which finally implies that~\eqref{eq:final-almost-there}\ is equal to%
\begin{multline}
\exp\!\left(  -\frac{1}{4}w_{1}^{T}\Omega^{T}V_{\mathcal{N}(\sigma)}\Omega
w_{1}-\frac{1}{4}\left(  \Omega^{T}w_{2}\right)  ^{T}V_{\sigma}\Omega^{T}%
w_{2}-\frac{1}{2}w_{1}^{T}\Omega^{T}X\sqrt{I+\left(  V_{\sigma}\Omega\right)
^{-2}}V_{\sigma}\Omega^{T}w_{2}\right) \label{eq:final-final-RHS-Petz}\\
=\exp\!\left(  -\frac{1}{4}\left(  \Omega^{T}w\right)  ^{T}V_{\mathcal{N}%
(\sigma)}\Omega^{T}w_{1}-\frac{1}{4}w_{2}^{T}\Omega^{T}V_{\sigma}\Omega
w_{2}-\frac{1}{2}\left(  \Omega^{T}w_{1}\right)  ^{T}X\sqrt{I+\left(
V_{\sigma}\Omega\right)  ^{-2}}V_{\sigma}\Omega w_{2}\right)  .
\end{multline}
Comparing~\eqref{eq:final-final-RHS-Petz}\ with
\eqref{eq:final-final-LHS-Petz}, we see that we have shown the equality in
\eqref{eq:zero-mean-Petz-eqs}, which concludes the proof once
Lemma~\ref{lemma sqrt sandwich} is established.

Before proving Lemma~\ref{lemma sqrt sandwich}, we recall the following
result. Although an analogous formula was already established by
\cite[Lemma~2]{Kholevo1972}
and
\cite[Section~III]{PS00}, we provide a self-contained proof for the sake
of completeness.

\begin{lemma}
[Square root of Gaussian states \cite{Kholevo1972,PS00}]\label{lemma sqrt G} Let $\sigma_{0}$ be a
Gaussian state with vanishing first moments, i.e., $s_{\sigma_{0}}=0$. Then its
uniquely defined square root $\sqrt{\sigma_{0}}$ is a trace class operator
given by
\begin{equation}
\sqrt{\sigma_{0}} = \left(  \det V_{\sqrt{\sigma_{0}}} \right)  ^{1/4}
\int\frac{d^{2n} w}{(2\pi)^{n}} e^{-\frac14 w^{T} V_{\sqrt{\sigma_{0}}} w}
D_{\Omega w} , \label{sqrt G}%
\end{equation}
where $V_{\sqrt{\sigma_{0}}}$ is given by~\eqref{CM sqrt G}.
\end{lemma}

\begin{proof}
Call $K$ the right hand side of~\eqref{sqrt G}. Since the square root is
uniquely defined, it suffices to show that $K^{2}=\sigma_{0}$. In the
following we will use the shorthand $U\equiv V_{\sqrt{\sigma_{0}}}>0$, where
the strict positivity can be readily verified using~\eqref{CM sqrt G}, and is
also a consequence of $U$ being a legitimate quantum covariance matrix. We
obtain
\begin{align}
K^{2} &  =\left(  \left(  \det U\right)  ^{1/4}\int\frac{d^{2n}w}{(2\pi)^{n}%
}e^{-\frac{1}{4}w^{T}Uw}D_{\Omega w}\right)  ^{2}\\
&  =\left(  \det U\right)  ^{1/2}\int\frac{d^{2n}w\,d^{2n}z}{(2\pi)^{2n}%
}e^{-\frac{1}{4}w^{T}Uw-\frac{1}{4}z^{T}Uz}D_{\Omega w}D_{\Omega z}\\
&  =\left(  \det U\right)  ^{1/2}\int\frac{d^{2n}w\,d^{2n}z}{(2\pi)^{2n}%
}e^{-\frac{1}{4}w^{T}Uw-\frac{1}{4}z^{T}Uz-\frac{i}{2}w^{T}\Omega z}%
D_{\Omega(w+z)}.
\end{align}
Let us introduce the new variables $x\equiv w+z$ and $y\equiv\frac{w-z}{2}$,
in terms of which we obtain
\begin{align}
K^{2} &  =\left(  \det U\right)  ^{1/2}\int\frac{d^{2n}x\,d^{2n}y}{(2\pi
)^{2n}}e^{-\frac{1}{8}x^{T}Ux-\frac{1}{2}\left(  y+\frac{i}{2}U^{-1}\Omega
x\right)  ^{T}U\left(  y+\frac{i}{2}U^{-1}\Omega x\right)  -\frac{1}{8}%
x^{T}\Omega^{T}U^{-1}\Omega x}D_{\Omega x}\\
&  =\left(  \det U\right)  ^{1/2}\int\frac{d^{2n}x}{(2\pi)^{n}}e^{-\frac{1}%
{8}x^{T}\left(  U+\Omega^{T}U^{-1}\Omega\right)  x}\left(  \int\frac{d^{2n}%
y}{(2\pi)^{n}}e^{-\frac{1}{2}\left(  y+\frac{i}{2}U^{-1}\Omega x\right)
^{T}U\left(  y+\frac{i}{2}U^{-1}\Omega x\right)  }\right)  D_{\Omega x}\\
&  =\left(  \det U\right)  ^{1/2}\int\frac{d^{2n}x}{(2\pi)^{n}}e^{-\frac{1}%
{4}x^{T}Vx}\left(  \int\frac{d^{2n}\tilde{y}}{(2\pi)^{n}}e^{-\frac{1}{2}%
\tilde{y}^{T}U\tilde{y}}\right)  D_{\Omega x}\\
&  =\int\frac{d^{2n}x}{(2\pi)^{n}}e^{-\frac{1}{4}x^{T}Vx}D_{\Omega x}\\
&  =\sigma_{0},
\end{align}
where we defined the shifted variable $\tilde{y}\equiv y+\frac{i}{2}%
U^{-1}\Omega x$ to perform the internal Gaussian integral and in the last step
we appealed to the representation~\eqref{G state 2}. Moreover, in the above calculation
we observed that
\begin{align}
\Omega^{T}U^{-1}\Omega &  =\Omega U^{-1}\Omega^{T}=\left(  U\Omega\right)
^{-1}\Omega\label{average U = V eq1}\\
&  =\left(  \sqrt{I+(V_{\sigma}\Omega)^{-2}}V_{\sigma}\Omega+V_{\sigma}%
\Omega\right)  ^{-1}\Omega\label{average U = V eq2}\\
&  =\left(  \sqrt{I+(V_{\sigma}\Omega)^{-2}}V_{\sigma}\Omega-V_{\sigma}%
\Omega\right)  \Omega\label{average U = V eq3}\\
&  =V_{\sigma}-\sqrt{I+(V_{\sigma}\Omega)^{-2}}V_{\sigma}%
\label{average U = V eq4}%
\end{align}
and hence
\begin{equation}
U+\Omega^{T}U^{-1}\Omega=2V_{\sigma}.\label{average U = V}%
\end{equation}
This concludes the proof of Lemma~\ref{lemma sqrt G}.
\end{proof}

\begin{lemma}
\label{lemma sqrt sandwich} Let $\sigma_{0}$ be a Gaussian state with
vanishing first moments $s_{\sigma_{0}}=0$. Then for all $x,y\in
\mathds{R}^{2n}$ we have
\begin{align}
\chi_{\sqrt{\sigma_{0}} D_{x} \sqrt{\sigma_{0}}}(y)  &  = \operatorname{Tr}%
\,\left[  D_{-y} \sqrt{\sigma_{0}} D_{x} \sqrt{\sigma_{0}} \right] \\
&  = \exp\! \left(  -\frac14 x^{T} \Omega^{T} V_{\sigma}\Omega x - \frac14
y^{T} \Omega^{T} V_{\sigma}\Omega y + \frac12 x^{T} \Omega^{T} \sqrt{I +
(V_{\sigma}\Omega)^{-2}} V_{\sigma}\Omega y \right)  .
\end{align}

\end{lemma}

\begin{proof}
To perform the computation, we just need to employ: (i) the representation
\eqref{sqrt G} for the square root of a Gaussian state with zero mean, (ii)
the composition identity~\eqref{comp displacement}, (iii) the orthogonality
relation~\eqref{trace displacement}; and (iv) the standard formula for a
Gaussian integral, i.e.
\begin{equation}
\int\frac{d^{2n}z}{(2\pi)^{n}}e^{-\frac{1}{2}z^{T}Uz+\frac{1}{2}a^{T} U z}%
=\frac{e^{\frac{1}{8}a^{T}Ua}}{\sqrt{\det U}},\label{Gaussian int}%
\end{equation}
valid for $U>0$. Defining again $U\equiv V_{\sqrt{\sigma_{0}}}$, we obtain
\begin{align}
\chi_{\sqrt{\sigma_{0}}D_{x}\sqrt{\sigma_{0}}}(y) &  =\operatorname{Tr}%
\,\left[  D_{-y}\sqrt{\sigma_{0}}D_{x}\sqrt{\sigma_{0}}\right]  \\
&  =\left(  \det U\right)  ^{1/2}\int\frac{d^{2n}w\,d^{2n}z}{(2\pi)^{2n}}%
\exp\!\left(  -\frac{1}{4}w^{T}Uw-\frac{1}{4}z^{T}Uz\right)  \operatorname{Tr}%
\,[D_{-y}D_{\Omega w}D_{x}D_{\Omega z}]\\
&  =\left(  \det U\right)  ^{1/2}\int\frac{d^{2n}w\,d^{2n}z}{(2\pi)^{2n}}%
\exp\!\left(  -\frac{1}{4}w^{T}Uw-\frac{1}{4}z^{T}Uz+\frac{i}{2}x^{T}%
z-\frac{i}{2}y^{T}w\right)  \nonumber\\
&  \qquad\quad\times\operatorname{Tr}\,[D_{\Omega w-y}D_{\Omega z+x}]\\
&  =\left(  \det U\right)  ^{1/2}\int\frac{d^{2n}w\,d^{2n}z}{(2\pi)^{2n}}%
\exp\!\left(  -\frac{1}{4}w^{T}Uw-\frac{1}{4}z^{T}Uz+\frac{i}{2}x^{T}%
z-\frac{i}{2}y^{T}w\right)  \nonumber\\
&  \qquad\quad\times(2\pi)^{n}\delta\left(  \Omega w-y+\Omega z+x\right)  \\
&  =\left(  \det U\right)  ^{1/2}\int\frac{d^{2n}w\,d^{2n}z}{(2\pi)^{2n}}%
\exp\!\left(  -\frac{1}{4}w^{T}Uw-\frac{1}{4}z^{T}Uz+\frac{i}{2}x^{T}%
z-\frac{i}{2}y^{T}w\right)  \nonumber\\
&  \qquad\quad\times(2\pi)^{n}\delta\left(  w-\Omega(x-y)+z\right)  \\
&  =\left(  \det U\right)  ^{1/2}\int\frac{d^{2n}z}{(2\pi)^{n}}\exp\!\left(
-\frac{1}{4}(\Omega(x-y)-z)^{T}U(\Omega(x-y)-z)\right.  \nonumber\\
&  \qquad\quad\left.  -\frac{1}{4}z^{T}Uz+\frac{i}{2}x^{T}z-\frac{i}{2}%
y^{T}(\Omega(x-y)-z)\right)  \\
&  =\left(  \det U\right)  ^{1/2}\exp\!\left(  -\frac{1}{4}(x-y)^{T}\Omega
^{T}U\Omega(x-y)+\frac{i}{2}x^{T}\Omega y\right)  \nonumber\\
&  \qquad\quad\times\int\frac{d^{2n}z}{(2\pi)^{n}}\exp\!\left(  -\frac{1}%
{2}z^{T}Uz+\frac{1}{2}\left(  \Omega(x-y)+iU^{-1}(x+y)\right)  ^{T}Uz\right)
\\
&  =\exp\!\left(  -\frac{1}{4}(x-y)^{T}\Omega^{T}U\Omega(x-y)+\frac{i}{2}%
x^{T}\Omega y\right)  \nonumber\\
&  \qquad\quad\times\exp\!\left(  \frac{1}{8}\left(  \Omega(x-y)+iU^{-1}%
(x+y)\right)  ^{T}U\left(  \Omega(x-y)+iU^{-1}(x+y)\right)  \right)  \\
&  =\exp\!\left(  -\frac{1}{4}x^{T}\Omega^{T}\frac{U+\Omega U^{-1}\Omega^{T}%
}{2}\Omega x-\frac{1}{4}y^{T}\Omega^{T}\frac{U+\Omega U^{-1}\Omega^{T}}%
{2}\Omega y\right.  \nonumber\\
&  \qquad\quad\left.  +\frac{1}{2}x^{T}\Omega^{T}\frac{U-\Omega U^{-1}%
\Omega^{T}}{2}\Omega y\right)  \\
&  =\exp\!\left(  -\frac{1}{4}x^{T}\Omega^{T}V_{\sigma}\Omega x-\frac{1}%
{4}y^{T}\Omega^{T}V_{\sigma}\Omega y+\frac{1}{2}x^{T}\Omega^{T}\sqrt
{I+(V_{\sigma}\Omega)^{-2}}V_{\sigma}\Omega y\right)  .
\end{align}
In the last step, we used~\eqref{average U = V} and the analogous relation
$U-\Omega U^{-1}\Omega^{T}=2\sqrt{I+(V_{\sigma}\Omega)^{-2}}V_{\sigma}$,
deduced again with the help of~\eqref{average U = V eq4}.
\end{proof}

\subsection{Step 4:\ The Gaussian Petz map satisfies the Petz equations for
all bounded operators}

\label{sec rigorous}

Throughout Section~\ref{subsec Hilbert-Schmidt}, we showed that the Petz
equation in~\eqref{eq:Petz-equations} is satisfied by the Gaussian channel in
\eqref{Petz cm} for all Hilbert--Schmidt operators. In this section we
complete the argument by showing that the same is true for all bounded
operators $A,B$ in~\eqref{eq:Petz-equations}. Thus, as a consequence of the
development in this section, we can conclude from a result of
\cite{petz86,petz88,PETZ}\ that the Gaussian channel in~\eqref{Petz cm} is in
fact the Petz map for $\sigma$ and $\mathcal{N}$.

The argument given here is standard, but we provide it here for completeness.
Proceeding, we have to show that the following Petz equation%
\begin{equation}
\langle A,\mathcal{N}^{\dag}(B)\rangle_{\sigma}\,=\,\langle\mathcal{P}^{\dag
}(A),B\rangle_{\mathcal{N}(\sigma)}\label{Petz eq}%
\end{equation}
is satisfied for all bounded $A,B$, supposing that we can verify it only for a
restricted class of $A,B$, for instance, those which are finite-rank (note
that finite-rank operators are Hilbert--Schmidt). Recall that a sequence
$(T_{n})_{n\in\mathds{N}}$ of operators on a Hilbert space $\mathcal{H}$ is
said to be weakly convergent to $T$, and we write $T_{n}\xrightarrow{w}T$, if
\begin{equation}
\lim_{n\rightarrow\infty}\langle\alpha,T_{n}\beta\rangle\,=\,\langle
\alpha,T\beta\rangle\qquad\forall\ \alpha,\beta\in\mathcal{H}%
\,.\label{weak conv}%
\end{equation}

We start by recalling the well-known fact that finite-rank operators are
weakly dense in the set of bounded operators. It is straightforward to show
this for all bounded $A$: one has $\Pi_{n}A\Pi_{n}\xrightarrow{w}A$, with
$\Pi_{n}$ denoting the projector onto the first $n$ vectors of the canonical
basis. Indeed, taking arbitrary vectors $\alpha,\beta\in\mathcal{H}$, we have
that%
\begin{align}
\big|\langle\Pi_{n}\alpha,A\Pi_{n}\beta\rangle-\langle\alpha,A\beta
\rangle\big| &  =\big|\langle\Pi_{n}\alpha-\alpha,A\beta\rangle+\langle\Pi
_{n}\alpha,A(\Pi_{n}\beta-\beta)\rangle\big|\\
&  \leq\big|\langle\Pi_{n}\alpha-\alpha,A\beta\rangle\big|+\big|\langle\Pi
_{n}\alpha,A(\Pi_{n}\beta-\beta)\rangle\big|\\
&  \leq\Vert A\Vert_{\infty}\big(\Vert\alpha\Vert\,\Vert\Pi_{n}\beta
-\beta\Vert+\Vert\beta\Vert\,\Vert\Pi_{n}\alpha-\alpha\Vert
\big)\xrightarrow[n\rightarrow\infty]{}0\,.
\end{align}

An important tool in our discussion will be the \textit{uniform boundedness
principle}~\cite{B87}, which states that if a sequence of operators
$(T_{n})_{n\in\mathds{N}}$ is such that the sequence of norms $\left(  \Vert
T_{n}\alpha\Vert\right)  _{n\in\mathds{N}}$ is bounded for all $\alpha
\in\mathcal{H}$, then the sequence of operator norms $\Vert T_{n}\Vert
_{\infty}$ is itself bounded.

\begin{lemma}
\label{app lemma 1} Let $(A_{n})_{n\in\mathds{N}}$ be a weakly convergent
sequence of operators. Then the sequence of operator norms $(\|A_{n}%
\|_{\infty})_{n\in\mathds{N}}$ is bounded.
\end{lemma}

\begin{proof}
Pick an arbitrary $\alpha\in\mathcal{H}$, and consider the sequence of
functionals $f_{n}^{(\alpha)}:\mathcal{H}\rightarrow\mathds{C}$ acting as
$f_{n}^{(\alpha)}(\beta)=\langle A_{n}\alpha, \beta\rangle$. Since
$(A_{n})_{n\in\mathds{N}}$ is weakly convergent, $f_{n}^{(\alpha)}(\beta)$ has
a limit in $\mathds{C}$, and in particular it is bounded. Since this holds for
all $\beta$, the uniform boundedness principle states that the norms
$\|f_{n}^{(\alpha)}\|_{\infty}=\|A_{n}\alpha\|$ must be bounded as well. Since
this holds for an arbitrary $\alpha$, another application of the uniform
boundedness principle guarantees that also $(\|A_{n}\|_{\infty})_{n\in
\mathds{N}}$ is bounded.
\end{proof}

\bigskip

Now we discuss some alternative definitions of weak convergence.

\begin{lemma}
\label{app lemma 2} Given a sequence $(A_{n})_{n\in\mathds{N}}$ of operators
on a Hilbert space, the following are equivalent:

\begin{enumerate}
\item $A_{n}\xrightarrow{w} A$;

\item $\operatorname{Tr}[ \rho A_{n} ]\rightarrow\operatorname{Tr}[ \rho A]$
for all states $\rho$;

\item $\operatorname{Tr}[ Z A_{n} ]\rightarrow\operatorname{Tr}[ ZA]$ for all
trace-class $Z$.
\end{enumerate}
\end{lemma}

\begin{proof}
\begin{itemize}
\item[$1.\!\Rightarrow\!2.$] Since $A_{n}-A\xrightarrow{w}0$, Lemma
\ref{app lemma 1} ensures that there is a constant $M$ such that for
sufficiently large $n$ $\Vert A_{n}-A\Vert_{\infty}\leq M<\infty$. Since
$\rho$ is a state, for all $\varepsilon>0$ we can fix a projector $\Pi$ onto a
finite-dimensional subspace such that $\Vert\rho-\Pi\rho\Pi\Vert_{1}\leq
\frac{\varepsilon}{2M}$. Moreover, the weak convergence of $A_{n}$ and the
fact that $\Pi\rho\Pi$ has finite support imply that $\operatorname{Tr}%
[\Pi\rho\Pi(A_{n}-A)]<\frac{\varepsilon}{2}$ for sufficiently large $n$. Then
\begin{align}
\big|\operatorname{Tr}[\rho(A_{n}-A)]\big|  &  \leq\big|\operatorname{Tr}%
[\Pi\rho\Pi(A_{n}-A)]\big|+\big|\operatorname{Tr}[(\rho-\Pi\rho\Pi
)(A_{n}-A)]\big|\\
&  \leq\frac{\varepsilon}{2}+\Vert\rho-\Pi\rho\Pi\Vert_{1}\Vert A_{n}%
-A\Vert_{\infty}\leq\frac{\varepsilon}{2}+\frac{\varepsilon}{2M}%
\,M=\varepsilon\,,
\end{align}
for sufficiently large $n$. This shows that $\operatorname{Tr}[\rho
A_{n}]\rightarrow\operatorname{Tr}[\rho A]$.

\item[$2.\!\Rightarrow\!3.$] This follows directly because all trace-class
operators can be written as a complex linear combination of four states.

\item[$3.\!\Rightarrow\!1.$] This implication becomes clear once we choose $Z$
to be the rank-one operator $Zx\equiv(\alpha,x)\beta$ and apply the definition
of weak convergence~\eqref{weak conv}.
\end{itemize}

This concludes the proof.
\end{proof}

\begin{corollary}
\label{app lemma 3} Let $\mathcal{N}$ be a quantum channel. If a sequence of
bounded operators $A_{n}$ satisfies $A_{n}\xrightarrow{w}A$, then
$\mathcal{N}^{\dag}(A_{n})\xrightarrow{w}\mathcal{N}^{\dag}(A)$.
\end{corollary}

\begin{proof}
We verify condition 2 of Lemma~\ref{app lemma 2}. Pick a state $\rho$. One
has
\begin{equation}
\operatorname{Tr}[ \rho\mathcal{N}^{\dag}(A_{n})] = \operatorname{Tr}[
\mathcal{N}(\rho) A_{n}] \rightarrow\operatorname{Tr}[ \mathcal{N}(\rho) A] =
\operatorname{Tr}[ \rho\mathcal{N}^{\dag}(A)]\, ,
\end{equation}
where we used again condition 2 of Lemma~\ref{app lemma 2} in order to take
the limit.
\end{proof}

\bigskip

Now we come to our decisive tool:

\begin{corollary}
\label{app lemma 4} Let $A_{n}\xrightarrow{w} A$ be a weakly convergent
sequence of operators. Then the following holds for an arbitrary state
$\omega$ and bounded operator $B$:
\begin{equation}
\langle A_{n}, B \rangle_{\omega} \xrightarrow[n\rightarrow\infty]{} \langle
A, B\rangle_{\omega}.
\end{equation}

\end{corollary}

\begin{proof}
It suffices to note that $\omega^{1/2}B\omega^{1/2}$ is a trace-class
operator. Applying condition 3 of Lemma~\ref{app lemma 2} yields the statement.
\end{proof}

\begin{theorem}
If the Petz equation in~\eqref{Petz eq} is satisfied for all finite-rank
operators $A,B$, then the same is true for all bounded $A,B$.
\end{theorem}

\begin{proof}
For bounded $A,B$, consider sequences of finite-rank operators $(A_{n}%
)_{n\in\mathds{N}}$ and $(B_{m})_{m\in\mathds{N}}$ such that $A_{n}%
\xrightarrow{w}A$ and $B_{m}\xrightarrow{w}B$. Then we have
\begin{equation}
\langle A_{n},\mathcal{N}^{\dag}(B_{m})\rangle_{\sigma}\,=\,\langle
\mathcal{P}^{\dag}(A_{n}),B_{m}\rangle_{\mathcal{N}(\sigma)}\,.
\end{equation}
Since Corollary~\ref{app lemma 3} implies that $\mathcal{N}^{\dag}%
(B_{m})\xrightarrow{w}\mathcal{N}^{\dag}(B)$, we can safely use
Corollary~\ref{app lemma 4} to take the limit $m\rightarrow\infty$ on both
sides, which yields
\begin{equation}
\langle A_{n},\mathcal{N}^{\dag}(B)\rangle_{\sigma}\,=\,\langle\mathcal{P}%
^{\dag}(A_{n}),B\rangle_{\mathcal{N}(\sigma)}\,.
\end{equation}
With the same argument we can now take the limit $n\rightarrow\infty$, and
this concludes the proof.
\end{proof}

\section{Conclusion}

\label{sec:conclusion}The main result of this paper is Theorem~\ref{thm:main},
which establishes an explicit form for the Petz map as a bosonic Gaussian
channel whenever the state $\sigma$ and the channel $\mathcal{N}$ are bosonic
Gaussian. Our proof approach was first to consider three ansatzes in order to
arrive at a hypothesis for the Gaussian form of the Petz map. These ansatzes
included 1)\ working with the form of the Petz map in~\eqref{Petz} in spite of
the fact that $\left[  \mathcal{N}(\sigma)\right]  ^{-1}$ is an unbounded
operator, 2)\ negating the covariance matrix of the Gaussian state $\sigma$ if
$\sigma$ is inverted, and 3)\ assuming that the $X$ matrix in
\eqref{eq:Gaussian-channel-char-func}, corresponding to a Gaussian channel, is
invertible. After deducing a hypothesis for an explicit form, we proved that
this hypothesis is in fact correct, by demonstrating that the Gaussian Petz
channel satisfies the equations in~\eqref{Petz eq}\ for all bounded operators
$A$ and $B$. Additionally, our Appendix~\ref{golden}, building on
\cite[Equation (30)]{Balian1969}, offers a powerful tool for computing
products of exponentials of inhomogeneous quadratic Hamiltonians. We suspect
that the ideas and tools presented in this paper will be useful for making
future progress in Gaussian quantum information.

As an immediate application of our results, we can consider whether the strongest form of the conjecture in~\eqref{conj VV} holds for all Gaussian states and channels, with the recovery channel taken to be the Petz recovery map. This question stems from an intuition that these Gaussian objects behave somewhat more classically than arbitrary quantum states or channels, and it is known that the conjecture in~\eqref{conj VV} holds for classical probability distributions. However, this intuition turns out to be fallacious: simple numerical searches yield plenty of counterexamples. A Mathematica file to generate and check such counterexamples is included in our arXiv post \cite{LDW17}. We stress that these numerical tests have been made possible by the fact that we exhibited an explicit formula for the action of the Petz recovery map.

In future work, it would be interesting to determine whether the following
inequality, considered in~\cite{BSW14,SBW14}, could be satisfied whenever all
of the objects involved are Gaussian:%
\begin{equation}
D(\rho\Vert\sigma)\geq D(\mathcal{N}(\rho)\Vert\mathcal{N}(\sigma))-\log
F(\rho,(\mathcal{P}_{\sigma,\mathcal{N}}\circ\mathcal{N})(\rho)).
\end{equation}
More generally, one could consider the various inequalities proposed in
\cite{BSW15a} for the Gaussian case.

\textit{Note}: Our results originally appeared on the arXiv as \cite{LDW17}. We remark here that another work presented the Gaussian Petz recovery map built from zero-mean states and channels \cite{Beny2017}, by making use of methods discussed in \cite{Beny17a}.

We thank Gerardo Adesso, Prabha Mandayam, Alessio Serafini, Kaushik
Seshadreesan, and Andreas Winter for discussions related to this paper.
Furthermore, LL thanks Davide Orsucci for his contribution to the proof
contained in Section~\ref{sec rigorous}. LL acknowledges financial support
from the European Research Council (AdG IRQUAT No. 267386), the Spanish
MINECO (Project no. FIS2013-40627-P and no. FIS2016-86681-P), and the Generalitat de Catalunya (CIRIT
Project no. 2014 SGR 966). SD acknowledges support from the Economic
Development Assistantship of Louisiana\ State University. MMW\ acknowledges
support from the National Science Foundation under Award No.~1714215.

\appendix

\section{Golden rule to handle exponentials of inhomogeneous quadratic
Hamiltonians}

\label{golden}

Very often in quantum optics one has to manipulate products of exponentials of
(inhomogeneous) quadratic Hamiltonians, i.e., operators of the following
form:
\begin{equation}
\mathcal{H} = \frac{i}{2} r^{T}\Omega X r + is^{T}\Omega r +\frac{i}{2} a .
\label{q hamilt}%
\end{equation}
For instance, a typical task consists in turning such a product into a single
exponential of another quadratic operator of the same form. In the above
equation, $r=(x_{1},\ldots, x_{n},p_{1},\ldots, p_{n})^{T}$ denotes the column
vector of canonical coordinates, $[r,r^{T}]=i\Omega$ and $\Omega X$ can be
assumed to be symmetric. Within the context of quantum optics, several methods
have been developed to deal with such calculations, which can be very involved
otherwise. In particular, a general formula for converting product of
exponentials of quadratic operators into a single exponential has been found
in~\cite[Equation (30)]{Balian1969}.

The main idea behind the approach we discuss here is not particularly novel
and has been already successfully exploited in quantum optics. For a thorough
review with many examples, we refer the reader to~\cite[Chapter 2]%
{puri2001mathematical}. However, the particular example we present does not
seem to have been considered before, and we believe it is of practical
importance to make the kind of computations we performed in this paper much
easier and more intuitive. To demonstrate the convenience of our method, we
conclude this appendix with an alternative proof of Lemma
\ref{lemma sqrt sandwich}.

The starting point is the observation that quadratic operators of the form
\eqref{q hamilt} form a Lie algebra, a fact which is easily seen to be a
consequence of the canonical commutation relations~\eqref{CCR}. Namely, it is
easy to see that
\begin{align}
\left[  \frac{i}{2} r^{T}\Omega X r + i s^{T} \Omega r + \frac{i}{2}
a,\ \frac{i}{2} r^{T}\Omega Y r + i t^{T} \Omega r + \frac{i}{2} b \right]  =
\frac{i}{2} r^{T} \Omega[X,Y] r + i (Xt - Ys)^{T} \Omega r - i s^{T} \Omega t
. \label{comm H}%
\end{align}
As is well-known, given $\mathcal{H}_{1},\mathcal{H}_{2}$ of the form
\eqref{q hamilt}, the operator $\mathcal{H}_{3}$ satisfying
\begin{equation}
e^{\mathcal{H}_{1}} e^{\mathcal{H}_{2}}\, =\, e^{\mathcal{H}_{3}}%
\end{equation}
depends only on the Lie algebra generated by $\mathcal{H}_{1}$ and
$\mathcal{H}_{2}$. Therefore, if we could construct an isomorphism turning the
Lie algebra of quadratic Hamiltonians into a (low-dimensional) matrix algebra,
we would be able to compute $\mathcal{H}_{3}$ as follows:

\begin{itemize}
\item[(i)] associate matrices $M_{1},M_{2}$ to $\mathcal{H}_{1},\mathcal{H}%
_{2}$ through the above isomorphism;

\item[(ii)] compute the Lie algebra element $M_{3}$ such that $e^{M_{1}}
e^{M_{2}}=e^{M_{3}}$;

\item[(iii)] use one last time the isomorphism to translate $M_{3}$ back to a
quadratic Hamiltonian $\mathcal{H}_{3}$.
\end{itemize}

It turns out that such an isomorphism can be found. An explicit example is as
follows:
\begin{equation}
\frac{i}{2} r^{T}\Omega X r + is^{T}\Omega r +\frac{i}{2} a\quad
\longleftrightarrow\quad%
\begin{pmatrix}
0 & s^{T} \Omega^{T} & a\\
0 & X & s\\
0 & 0 & 0
\end{pmatrix}
. \label{isomorphism}%
\end{equation}
The matrix Lie algebra we will be concerned about is thus formed by matrices
of the above form, with the only restriction that $\Omega X$ is symmetric. As
expected, the commutator between two such matrices takes the form
\begin{equation}
\left[
\begin{pmatrix}
0 & s^{T} \Omega^{T} & a\\
0 & X & s\\
0 & 0 & 0
\end{pmatrix}
,\
\begin{pmatrix}
0 & t^{T} \Omega^{T} & b\\
0 & Y & s\\
0 & 0 & 0
\end{pmatrix}
\right]  \ =\
\begin{pmatrix}
0 & (Xt-Ys)^{T} \Omega^{T} & - 2 s^{T}\Omega t\\
0 & [X,Y] & Xt-Ys\\
0 & 0 & 0
\end{pmatrix}
,
\end{equation}
mimicking~\eqref{comm H}. In order to apply our strategy, we need to compute
the exponential of a matrix belonging to our Lie algebra. It is an elementary
exercise to show that
\begin{equation}
\label{exp Lie}\exp\left[
\begin{pmatrix}
0 & s^{T} \Omega^{T} & a\\
0 & X & s\\
0 & 0 & 0
\end{pmatrix}
\right]  \ =\
\begin{pmatrix}
1 & \left(  \frac{I-e^{-X}}{X}\,s\right)  ^{T} \Omega^{T} & a + s^{T}
\Omega\frac{X-\sinh X}{X^{2}}\, s\\
0 & e^{X} & \frac{e^{X}-I}{X}\, s\\
0 & 0 & 1
\end{pmatrix}
.\\[1ex]
\end{equation}

We conclude this appendix by presenting an alternative and perhaps more
intuitive derivation of Lemma~\ref{lemma sqrt sandwich} that makes use of the
Lie algebra isomorphism~\eqref{isomorphism}. The advantage of this proof is
basically that it turns the cumbersome sequence of Gaussian integrals we
performed in the main body into a sequence of $3\times3$ block-matrix multiplications.

\vspace{2ex}

\begin{proof}
[Alternative proof of Lemma~\ref{lemma sqrt sandwich}]As a preliminary step,
we deduce from~\eqref{V coth} the expression for $\sinh\left(  \frac{i\Omega
H_{\sigma}}{2}\right)  $. Since $i\Omega H_{\sigma}$ has real eigenvalues, we
can apply the identity $\sinh x = \frac{\coth(x)^{-1}}{\sqrt{1-\coth(x)^{2}}}$
(valid for real $x$) and~\eqref{V coth} to obtain
\begin{equation}
\sinh\left(  \frac{i\Omega H_{\sigma}}{2}\right)  = \frac{(iV\Omega)^{-1}%
}{\sqrt{I + (V\Omega)^{-2}}} . \label{sinh (i Omega H / 2)}%
\end{equation}

Now, let us show how to compute $\sqrt{\sigma_{0}}D_{x}\sqrt{\sigma_{0}}$ for
any given $x$. We can employ the exponential form of $\sigma_{0}$ as given in
\eqref{eq:exp-form-Gaussian}, which in our case becomes $\sigma_{0}=Z_{\sigma
}e^{-\frac{1}{2}r^{T}H_{\sigma}r}$. For the sake of simplicity, we ignore the
normalisation constant $Z_{\sigma}$ for the moment. Also, let us omit the
subscripts $\sigma$ throughout the calculation. We find
\begin{align}
\sqrt{\sigma_{0}}D_{x}\sqrt{\sigma_{0}} &  \propto e^{\frac{i}{4}r^{T}%
\Omega(-i\Omega H)r}e^{ix^{T}\Omega r}e^{\frac{i}{4}r^{T}\Omega(-i\Omega
H)r}\\
&  \overset{\mathclap{\scriptsize \mbox{(i)}}}{\longrightarrow}\exp%
\begin{pmatrix}
0 & 0 & 0\\
0 & -i\Omega H/2 & 0\\
0 & 0 & 0
\end{pmatrix}
\exp%
\begin{pmatrix}
0 & x^{T}\Omega^{T} & 0\\
0 & 0 & x\\
0 & 0 & 0
\end{pmatrix}
\exp%
\begin{pmatrix}
0 & 0 & 0\\
0 & -i\Omega H/2 & 0\\
0 & 0 & 0
\end{pmatrix}
\\
&  \overset{\mathclap{\scriptsize \mbox{(ii)}}}{=}%
\begin{pmatrix}
1 & 0 & 0\\
0 & e^{-i\Omega H/2} & 0\\
0 & 0 & 1
\end{pmatrix}%
\begin{pmatrix}
1 & x^{T}\Omega^{T} & 0\\
0 & I & x\\
0 & 0 & 1
\end{pmatrix}%
\begin{pmatrix}
1 & 0 & 0\\
0 & e^{-i\Omega H/2} & 0\\
0 & 0 & 1
\end{pmatrix}
\\
&  =%
\begin{pmatrix}
1 & x^{T}\Omega^{T}e^{-i\Omega H/2} & 0\\
0 & e^{-i\Omega H} & e^{-i\Omega H/2}x\\
0 & 0 & 1
\end{pmatrix}
\end{align}%
\begin{align}
&  \overset{\mathclap{\scriptsize \mbox{(iii)}}}{=}%
\begin{pmatrix}
1 & \left(  \frac{1}{2\sinh\left(  \frac{i\Omega H}{2}\right)  }x\right)
^{T}\Omega^{T} & 0\\
0 & I & \frac{1}{2\sinh\left(  \frac{i\Omega H}{2}\right)  }x\\
0 & 0 & 1
\end{pmatrix}%
\begin{pmatrix}
1 & 0 & \frac{1}{4}x^{T}\Omega\frac{e^{-i\Omega H}}{\sinh\left(  \frac{i\Omega
H}{2}\right)  ^{2}}x\\
0 & e^{-i\Omega H} & 0\\
0 & 0 & 1
\end{pmatrix}
\nonumber\\
&  \qquad\qquad\times%
\begin{pmatrix}
1 & \left(  -\frac{1}{2\sinh\left(  \frac{i\Omega H}{2}\right)  }x\right)
^{T}\Omega^{T} & 0\\
0 & I & -\frac{1}{2\sinh\left(  \frac{i\Omega H}{2}\right)  }x\\
0 & 0 & 1
\end{pmatrix}
\\
&  \overset{\mathclap{\scriptsize \mbox{(iv)}}}{\longrightarrow}\exp\!\left(
i\left(  \frac{1}{2\sinh\left(  \frac{i\Omega H}{2}\right)  }x\right)
^{T}\Omega r\right)  \exp\!\left(  \frac{i}{2}r^{T}\Omega(-i\Omega
H)r+\frac{i}{8}x^{T}\Omega\frac{e^{-i\Omega H}}{\sinh\left(  \frac{i\Omega
H}{2}\right)  ^{2}}x\right)  \nonumber\\
&  \qquad\qquad\times\exp\!\left(  i\left(  -\frac{1}{2\sinh\left(
\frac{i\Omega H}{2}\right)  }x\right)  ^{T}\Omega r\right)
\end{align}%
\begin{align}
&  =\exp\!\left(  \frac{i}{8}x^{T}\Omega\frac{e^{-i\Omega H}}{\sinh\left(
\frac{i\Omega H}{2}\right)  ^{2}}x\right)  D_{\frac{1}{2\sinh\left(
\frac{i\Omega H}{2}\right)  }x}\ e^{-\frac{1}{2}r^{T}Hr}\ D_{-\frac{1}%
{2\sinh\left(  \frac{i\Omega H}{2}\right)  }x}\\
&  \overset{\mathclap{\scriptsize \mbox{(v)}}}{=}\exp\!\left(  -\frac{i}%
{8}x^{T}\Omega\frac{\sinh(i\Omega H)}{\sinh\left(  \frac{i\Omega H}{2}\right)
^{2}}x\right)  D_{\frac{1}{2\sinh\left(  \frac{i\Omega H}{2}\right)  }%
x}\ e^{-\frac{1}{2}r^{T}Hr}\ D_{-\frac{1}{2\sinh\left(  \frac{i\Omega H}%
{2}\right)  }x}\\
&  \overset{\mathclap{\scriptsize \mbox{(vi)}}}{=}\exp\!\left(  -\frac{1}%
{4}x^{T}\Omega^{T}V\Omega x\right)  D_{\frac{i}{2}\sqrt{I+(V\Omega)^{-2}%
}V\Omega x}\ e^{-\frac{1}{2}r^{T}Hr}\ D_{-\frac{i}{2}\sqrt{I+(V\Omega)^{-2}%
}V\Omega x}.
\end{align}
The justification of these steps is as follows: (i) forward application of the
isomorphism~\eqref{isomorphism}; (ii) exponential formula~\eqref{exp Lie};
(iii) direct verification; (iv) backward application of the isomorphism
\eqref{isomorphism}; (v) we use $x^{T}Ax=\frac{1}{2}x^{T}(A+A^{T})x$ to
symmetrise the matrix inside the first exponential; (vi) we employ
\eqref{sinh (i Omega H / 2)}, the hyperbolic trigonometric identity
$\frac{\sinh y}{\sinh(y/2)^{2}}=2\coth(y/2)$ and~\eqref{V coth}. Once we
reintroduce the normalisation $Z_{\sigma}$, the above calculation shows that
\begin{equation}
\sqrt{\sigma_{0}}D_{x}\sqrt{\sigma_{0}}=e^{-\frac{1}{4}x^{T}\Omega
^{T}V_{\sigma}\Omega x}D_{\frac{i}{2}\sqrt{I+(V_{\sigma}\Omega)^{-2}}%
V_{\sigma}\Omega x}\ \sigma_{0}\ D_{-\frac{i}{2}\sqrt{I+(V_{\sigma}%
\Omega)^{-2}}V_{\sigma}\Omega x}.
\end{equation}
Then, using~\eqref{comp displacement} and~\eqref{char funct G} we see that
\begin{align}
\chi_{\sqrt{\sigma_{0}}D_{x}\sqrt{\sigma_{0}}}(y) &  =\operatorname{Tr}%
[D_{-y}\sqrt{\sigma_{0}}D_{x}\sqrt{\sigma_{0}}]\\
&  =e^{-\frac{1}{4}x^{T}\Omega^{T}V_{\sigma}\Omega x}\operatorname{Tr}\left[
D_{-y}D_{\frac{i}{2}\sqrt{I+(V_{\sigma}\Omega)^{-2}}V_{\sigma}\Omega x}%
\sigma_{0}D_{-\frac{i}{2}\sqrt{I+(V_{\sigma}\Omega)^{-2}}V_{\sigma}\Omega
x}\right]  \\
&  =e^{-\frac{1}{4}x^{T}\Omega^{T}V_{\sigma}\Omega x}\operatorname{Tr}\left[
D_{-\frac{i}{2}\sqrt{I+(V_{\sigma}\Omega)^{-2}}V_{\sigma}\Omega x}%
D_{-y}D_{\frac{i}{2}\sqrt{I+(V_{\sigma}\Omega)^{-2}}V_{\sigma}\Omega x}%
\sigma_{0}\right]  \\
&  =e^{-\frac{1}{4}x^{T}\Omega^{T}V_{\sigma}\Omega x}e^{\frac{1}{2}x^{T}%
\Omega^{T}\sqrt{I+(V_{\sigma}\Omega)^{-2}}V_{\sigma}\Omega y}\operatorname{Tr}%
\left[  D_{-y}\sigma_{0}\right]  \\
&  =e^{-\frac{1}{4}x^{T}\Omega^{T}V_{\sigma}\Omega x}e^{\frac{1}{2}x^{T}%
\Omega^{T}\sqrt{I+(V_{\sigma}\Omega)^{-2}}V_{\sigma}\Omega y}\chi_{\sigma_{0}%
}(y)\\
&  =\exp\!\left(  -\frac{1}{4}x^{T}\Omega^{T}V_{\sigma}\Omega x+\frac{1}%
{2}x^{T}\Omega^{T}\sqrt{I+(V_{\sigma}\Omega)^{-2}}V_{\sigma}\Omega y-\frac
{1}{4}y^{T}\Omega^{T}V_{\sigma}\Omega y\right)  ,
\end{align}
which concludes the proof.
\end{proof}

\section{Verifying complete positivity of the Gaussian Petz channel}

\label{sec:Petz-CP}One might want to verify explicitly the complete positivity
condition for the Petz map stated in Theorem~\ref{thm:main}, even if we know
from~\cite{petz86,petz88,PETZ} that~\eqref{Petz} has to be completely positive
by construction. Recall that a Gaussian channel defined by~\eqref{N cm} is
completely positive if and only if the inequality~\eqref{G cp condition} is met: i.e., if and only if ~\cite{adesso14,S17}
\begin{equation}
Y+i\Omega-iX\Omega X^{T}\geq0.\label{CP}
\end{equation}
We start with the following lemma:

\begin{lemma}
For all $V$ such that $V+i\Omega>0$, the following identity holds
\begin{equation}
\left(  I+ (\Omega V^{-2} \right)  ^{-1/2} V^{-1} (V+i\Omega) V^{-1} \left(
I+ (V\Omega)^{-2} \right)  ^{-1/2}\, =\, \frac{1}{V-i\Omega}.
\label{inversion identity}%
\end{equation}

\end{lemma}

\begin{proof}
This is a straightforward calculation after decomposing $V$ in the Williamson
form as $V=S(D\oplus D)S^{T}$, where $S$ is a symplectic matrix satisfying
$S\Omega S^{T}=\Omega$ and $D$ is a diagonal matrix of symplectic eigenvalues
(note that all entries of $D$ are larger than or equal to one).
\end{proof}

\vspace{2ex} With the above result in hand, we can write
\begin{align}
&  Y_{P}+i\Omega-iX_{P}\Omega X_{P}^{T}\,=\,V_{\sigma}+i\Omega-X_{P}\left(
V_{\mathcal{N}(\sigma)}+i\Omega\right)  X_{P}^{T}\,\nonumber\\
&  =\,\left(  I+(V_{\sigma}\Omega)^{-2}\right)  ^{1/2}V_{\sigma}\,\frac
{1}{V_{\sigma}-i\Omega}\,V_{\sigma}\left(  I+(\Omega V)^{-2}\right)
^{1/2}\,\nonumber\\
&  \qquad-\,\left(  I+\left(  V_{\sigma}\Omega\right)  ^{-2}\right)
^{1/2}V_{\sigma}X^{T}\left(  I+\left(  \Omega V_{\mathcal{N}(\sigma)}\right)
^{-2}\right)  ^{-1/2}V_{\mathcal{N}(\sigma)}^{-1}\left(  V_{\mathcal{N}%
(\sigma)}+i\Omega\right)  \nonumber\\
&  \qquad\times V_{\mathcal{N}(\sigma)}^{-1}\left(  I+\left(  V_{\mathcal{N}%
(\sigma)}\Omega\right)  ^{-2}\right)  ^{-1/2}XV_{\sigma}\left(  I+\left(
\Omega V_{\sigma}\right)  ^{-2}\right)  ^{1/2}\,\\
&  =\,\left(  I+(V_{\sigma}\Omega)^{-2}\right)  ^{1/2}V_{\sigma}\,\frac
{1}{V_{\sigma}-i\Omega}\,V_{\sigma}\left(  I+(\Omega V)^{-2}\right)
^{1/2}\,\nonumber\\
&  \qquad-\,\left(  I+\left(  V_{\sigma}\Omega\right)  ^{-2}\right)
^{1/2}V_{\sigma}X^{T}\frac{1}{V_{\mathcal{N}(\sigma)}-i\Omega}\,XV_{\sigma
}\left(  I+\left(  \Omega V_{\sigma}\right)  ^{-2}\right)  ^{1/2}\,\\
&  =\,\left(  I+\left(  V_{\sigma}\Omega\right)  ^{-2}\right)  ^{1/2}%
V_{\sigma}\left(  \frac{1}{V_{\sigma}-i\Omega}\,-\,X^{T}\frac{1}%
{V_{\mathcal{N}(\sigma)}-i\Omega}\,X\right)  V_{\sigma}\left(  I+\left(
\Omega V_{\sigma}\right)  ^{-2}\right)  ^{1/2}\,.\label{CP eq 1}%
\end{align}
Now, from $V_{\mathcal{N}(\sigma)}-i\Omega=XV_{\sigma}X^{T}+Y-i\Omega\geq
X(V_{\sigma}-i\Omega)X^{T}$ we obtain
\begin{equation}
\frac{1}{V_{\sigma}-i\Omega}\,-\,X^{T}\frac{1}{V_{\mathcal{N}(\sigma)}%
-i\Omega}\,X\,\geq\,\frac{1}{V_{\sigma}-i\Omega}\,-\,X^{T}\frac{1}%
{X(V_{\sigma}-i\Omega)X^{T}}\,X\,\geq\,0\,,\label{CP eq 2}%
\end{equation}
as it follows from the inequality $A^{-1}\geq X^{T}(XAX^{T})^{-1}X$, which is
in turn valid for all invertible $A$ and all matrices $X$ with no more rows
than columns and maximum rank. Plugging~\eqref{CP eq 2} into~\eqref{CP eq 1},
we conclude the condition in~\eqref{CP} for $X_{P}$ and $Y_{P}$, as desired.

\bibliographystyle{alpha}
\bibliography{Ref}

\end{document}